\documentclass[10pt, journal]{styles/IEEEtran}

\usepackage{epsfig}
\usepackage{amsthm}
\usepackage{amssymb} 
\usepackage[tbtags]{amsmath} 
\usepackage{graphics,eepic,epic}
\usepackage{latexsym}
\usepackage{euscript}
\usepackage{subfigure}
\usepackage{graphics,eepic,epic,psfrag}

\newcommand{\cX}{\mathcal{X}}

\newcommand{\cY}{\mathcal{Y}}
\newcommand{\cS}{\mathcal{S}}

\newcommand{\cK}{\mathcal{K}}

\newcommand{\bP}{\mathbf{P}}

\newcommand{\bR}{\mathbb{R}}

\newcommand{\fRe}{\mathfrak{R}}
\newcommand{\fIm}{\mathfrak{I}}
\newcommand{\bX}{\mathbf{X}}

\newcommand{\bx}{\mathbf{x}}

\newcommand{\cN}{\mathcal{N}}
\newcommand{\cCN}{\mathcal{CN}}

\newcommand{\cH}{\mathcal{H}}

\newcommand{\varNoise}{\sigma^2_W}

\newtheorem{theorem}{Theorem}
\newtheorem{definition}{Definition}
\newtheorem{lemma}{Lemma}

\usepackage{cite}

\begin{document}
% paper title
\title{Exploiting Channel Diversity in Secret Key Generation from Multipath Fading Randomness}
\pagenumbering{arabic}

\author{Yanpei Liu,~\IEEEmembership{Student Member,~IEEE,}
	 Stark~C.~Draper,~\IEEEmembership{Member,~IEEE,}
  	Akbar M.~Sayeed,~\IEEEmembership{Fellow,~IEEE} \thanks{Copyright (c) $2012$ IEEE. Personal use of this material is permitted. However, permission to use this material for any other purposes must be obtained from the IEEE by sending a request to pubs-permissions@ieee.org. }\thanks{This
    work was presented in part at the 45th annual Conference on
    Information Sciences and Systems (CISS), Baltimore MD, March
    2011.}  \thanks{The authors are with the Dept.~of Electrical and
    Computer Engineering, University of Wisconsin, Madison, WI 53706
    (E-mail: \{yliu73@wisc.edu, sdraper@ece.wisc.edu,
    akbar@engr.wisc.edu\}).}  \thanks{Y.~Liu and S.~C.~Draper were supported by the National Science Foundation under grant CCF-0963834 and by a grant from the Wisconsin Alumni Research Foundation.}}

\maketitle
\pagenumbering{arabic}

\begin{abstract}
We design and analyze a method to extract secret keys from the
randomness inherent to wireless channels. We study a channel model for
multipath wireless channel and exploit the channel diversity in generating secret key bits. We compare the key extraction methods based both on entire channel state information (CSI) and on single channel parameter such as the
received signal strength indicators (RSSI).  Due to the reduction in
the degree-of-freedom when going from CSI to RSSI, the rate of key
extraction based on CSI is far higher than that based on RSSI.  This
suggests that exploiting channel diversity and making CSI information available to higher layers would
greatly benefit the secret key generation.  We propose a key generation
system based on low-density parity-check (LDPC) codes and describe the
design and performance of two systems: one based on binary
LDPC codes and the other (useful at higher signal-to-noise ratios)
based on four-ary LDPC codes. 
\end{abstract}

\begin{keywords}
Common randomness, secret key generation, channel diversity, LDPC codes, Slepian-Wolf decoder
\end{keywords}

\section{Introduction}
\label{sec.introduction}

%% Yanpei, commented this out because it wasn't clear to me that our
%% technique addresses either of the failures of the current protocols
%% that you cite: denial-of-service or power.  If you have arguments
%% as to why our approach is superior in these dimensions then we
%% should add in those details.

%Current wireless communication security protocols
%are largely based on public key cryptography, such as Wired Equivalent
%Privacy (WEP), Extensible Authentication Protocol (EAP) and Wi-Fi
%Protected Access (WPA). However, the security of these techniques have
%been of great concern in recent years. For example, an attacker can
%cause a denial-of-service attack in a network equipped with
%WPA~\cite{Hao}. Also, recent research show that public key
%cryptography consumes a significant amount of computing resources and
%power. This places a significant load on the resources of circuitry of
%small-scale, especially battery-powered networks~\cite{Suman}.

In this paper we study the generation of secret keys based on the
inherent randomness of wireless multipath channels.  This study falls
into the broad area of physical layer security (see~\cite{Liang09} for
an overview of the area). In this setting the objective is for a pair
of users, generically referred to as {\em Alice} and {\em Bob}, to
extract a secret key from a naturally occurring source of randomness
observed by two users.  The central idea is that through a public
(i.e., {\em not} secret) discussion, Alice and Bob can de-noise their
correlated observations to generate, with high probability, a commonly
known string, which can serve as the key.  Of course, any eavesdropper
(typically named {\em Eve}) would use both her knowledge of the public
message and any observation she has to guess the key.  A source of
naturally occurring randomness that would be well suited to the key
generation application would be characterized by three properties.  It
would be easily and widely accessible, it would have a high level of
randomness, and it would be difficult for Eve to observe.  The
randomness inherent to wireless multipath fading channels, such as the
random amplitudes and phases of the channel response coefficients,
satisfies all three properties.

The ubiquity of personal wireless devices makes a multipath fading
channel an easily accessible, and hence very relevant, source of
randomness.  The fact that it has a high level of randomness and is
difficult to eavesdrop is due to the physics of electromagnetic wave
propagation.  In a rich multipath environment wireless channels have
high spatial and temporal variation.  For instance, whenever either
Alice or Bob moves, or whenever other scattering objects move between
them, the channel between them changes.  In terms of key extraction
this means that there is a continual influx of new randomness from
which to extract new and independent key bits.  For the same reason, an
eavesdropper that is listening on transmission between Alice and
Bob and that is even a few wavelengths away from either will observe a
nearly independent channel.  In terms of key extraction, this makes it
difficult to eavesdrop on the source of randomness (the channel
coefficients).

Modern wireless communication protocols typically use diversity signaling techniques 
such as orthogonal frequency-division multiplexing (OFDM) or Multiple Input Multiple Output (MIMO) antennas. These techniques exploit frequency, time and spatial diversity of the underlying wireless channel and improve the communication performance. By exploiting channel diversity in a similar manner in secret key generation one can harvest more randomness. Thus in this paper, we study an OFDM system as an example from the perspective of
secret key generation.  We characterize the suitableness of such
channels for key generation, both under the assumption of the
availability of full channel state information (CSI) and the
assumption of the availability of only received signal strength
indicators (RSSI). The latter is what is available to the higher
layers of existing wireless transceivers. We show that by exploiting the channel diversity in the CSI, one can significantly increase the rate at which the secret key bits can be generated relative to when channel diversity is not exploited (such as RSSI based method). Thus making CSI
available to the higher layers (where security is managed) in future
transceiver designs would greatly facilitate the adoption of the approach we propose.  We
also show that when extracting keys from CSI, one can, without loss of
rate, extracts key bits separately from the real and the imaginary
parts of each channel coefficient.  The same is not true for amplitude and
phase as there is correlation between the amplitudes and phases {\em
  across} two participating users.  We also detail an algorithm of the de-noising
needed in key extracted.  Our algorithm is based on low-density
parity-check (LDPC) codes.  We describe two designs.  One based
on binary and one based on non-binary (quaternary) LDPC codes.
Higher-alphabet codes are required to extract the full randomness of
the channel at higher signal-to-noise ratios (SNR).

There are many works in both theoretical analysis and practical implementation of physical layer security. Theoretical analysis in wire-tap channel date back to four decades ago \cite{TheWireTapChannel, BroadcastChannelswithConfidential} More recently, Bloch et al. propose the seminal practical opportunistic one-way secret key agreement protocol for Gaussian wiretap channel in \cite{BlochBarrosRodrigues}. The works done by Maurer \cite{maurer93} and Ahlswede and Csisz\'{a}r \cite{Ahlswede} show that correlated randomness can be used to generate secret keys. Their works lay down the analytical foundations for secret key generation in wireless communication. Sayeed and Perrig \cite{Sayeed} recognize the possibility of extracting secret keys from multipath randomness in wireless communication. Fundamental limits to key generation for multipath randomness are studied in ~\cite{Chou1, Chou3, Wilson, Ye_vtc07, Ye06, Nitinawarat07}.  In~\cite{Chou1, Chou3} the minimum energy-per-key-bit is characterized for rich fading channels and is extended in~\cite{Chou2} to sparse multipath channels. Eavesdropper with the ability to tamper the transmission has been studied by Maurer and Wolf \cite{MaurerWolf1, MaurerWolf2, MaurerWolf3}. More recently, Chou et al. study the secret key capacity of the sender-excited secret key agreement in \cite{ChouTanDraper}. Non-coherent secret key generation in which neither the sender nor the receiver have access to the channel state information has been studied in \cite{Ashish}. 

There are also many works on realizing physical layer security by designing practical secret key generating systems. These works are based on the earliest work by Hershey et al. \cite{HersheyHassan} and Hassan et al. \cite{HassanStark}. Ye et al. \cite{Ye1, Ye2} present an over-the-air implementation on $802.11$ platforms, prototyping a systematic design using a scalar fading channel coefficient. Jana et al. present yet another over-the-air implementation using the received signal strength indicators \cite{Suman}. Channel randomness is also exploited for device pairing \cite{ProxiMate} and authentication \cite{Liang1, Liang2, Banerjee}. Secret key generation system over MIMO has been considered in \cite{WallaceSharma} and the references therein. 

There are many related design related issues. Typical secret key generation process consists of three phases: randomness exploration, reconciliation and privacy amplification \cite{Ye2}. In randomness exploration, quantization is used to convert continuous observations to discretized information bits. A good quantizer should not only maximize the mutual information between Alice and Bob's bit sequences, but also reveal limited information to the eavesdropper. An algorithm is proposed in~\cite{Cardinal_1, Cardinal_2} to find such a quantizer. Ye et al. \cite{Ye2} propose an over-quantization technique to extract more bits per independent channel training. When the channel is over-static (long coherence time), filtering techniques, such as Discrete Cosine Transform in \cite{Yasukawa} and windowed moving average low pass filtering in \cite{Ye2}, are used to remove the redundancy in the extracted key bits.  Reconciliation process is typically done using various coding techniques, such as LDPC codes \cite{BlochBarrosRodrigues} and list-encoding \cite{ProxiMate}. For a detailed survey on reconciling two binary random variables, see \cite{BrassardSalvail}. Finally, universal hash functions are widely used \cite{Wilhelm, MaurerWolf3} for privacy amplification.  

In this paper, we show that the channel randomness can be further exploited through the channel diversity offered by the wireless front-end. We note that in many related works, such as \cite{Suman, Ye1, Ye2, ProxiMate}, secret key bits are extracted from a single parameter observed in wireless channel. This fundamentally limits the rate at which the secret key bits that can be extracted. For instance, in \cite{Ye1} only one bit can be extracted per independent channel realization although they over-quantize it to increase the number of bits in their later work \cite{Ye2}. Similarly, in \cite{ProxiMate} only one bit can be extracted per coherence time. We thus argue that by exploiting the channel diversity in wireless multipath fading channel, one can significantly improve the secret key capacity. 
%For RSSI-based methods such as \cite{Suman}, when the channel is over-static, meaning there
%is significant correlation across observations, key bits may contain
%long runs of 0's and 1's, which make it susceptible to potential
%attackers. In~\cite{Yasukawa}, an adaptive quantizer is proposed to
%address the problem caused by the over-static channel. They
%incorporate the Discrete Cosine Transform (DCT) transform to remove
%the redundancy caused by long runs of 0's and 1's. Universal hash
%functions (UHFs) are used in~\cite{Wilhelm} to amplify the randomness
%by extracting the maximum possible amount of entropy. In quantizer
%design, a good quantizer should not only maximize the mutual
%information between Alice and Bob's bit sequences, but also reveal
%limited information to the eavesdropper. An algorithm is proposed
%in~\cite{Cardinal_1, Cardinal_2} to find such a quantizer.

%~\cite{Suman, Sayeed, Wilhelm}.

%While physical layer CSI are estimated in modern transceivers, current
%standards do not provide higher layer with access to them. Many
%research therefore use received signal strength indicator (RSSI) as an
%alternative~\cite{Suman, Yasukawa}. However, RSSI does not represent
%the full CSI and only reports the received power which is used for
%power control monitoring. We demonstrate in this paper that secret key
%generation based on CSI significantly outperforms the one based on
%RSSI.

\subsection{System overview}
\label{sec.sysOverview}

To lend concreteness to the ensuing discussion we describe the
operation of the key extraction algorithm discussed later in the
paper.  To generate their correlated
observations Alice and Bob each transmits known channel sounding
(training) signals to each other.  This two-way training is done in two consecutive time slots.  As long as the
channel is static over these two time slots (the key assumption of our
model~\cite{Suman, Sayeed}) then, due to the reciprocity of
electromagnetic wave propagation, Alice and Bob both obtain (noisy)
observations of the same multipath fading channel coefficients.  Eve
is assumed to listen to both transmissions, but due to the fast
spatial decorrelation of multipath channels, we assume
for the remainder of the paper that her observations are independent
and thus useless for estimating the realized channel law.  Alice then
quantizes her observations into some finite alphabet. (If Alice did
not quantize her observations there would be no way Bob could recover
the {\em exact} same coefficients with high
probability.)  Alice then sends to Bob a public message.  In our algorithm the public message is the syndrome of some
length-$N$ error correcting code where $N$ is the length of Alice's
vector of quantized channel coefficients. Bob combines the public
message with his observations in his attempt to recover Alice's
quantized observations.  We describe two possibilities for Bob.
First, that he quantizes his own observations before his recovery
attempt (``hard'' decoding) and, second, that he bases his recovery
attempt on his un-quantized observations (strictly better ``soft''
decoding).

%\begin{figure}
 % \centerline{\epsfig{figure=Figure/system_diagram.pdf,width=9cm}}
%  \caption{System Diagram}
%  \label{fig.systematic_model}
%\end{figure}

We do not consider authentication in our proposed secret key
generation system \cite{Liang1, Liang2, Banerjee}. Therefore, our system does not address active attacks such as man-in-the-middle attacks. One
could always first authenticate the validity of Alice and Bob by using
public key cryptography before invoking our secret key generation
system. 

%%%%%\subsection{Our Contribution}
%We introduce an OFDM channel model from which we generate channel
%coefficients and introduce the ways to quantify the mutual information
%(and thus secret key capacity) between Alice and Bob. We discover that
%CSI based method is superior to RSSI based method proposed
%in~\cite{Suman} and others. Thus we suggest that modern standards
%should make CSI accessible to higher layers. Through simulation, we
%show that the loss due to Alice's quantization is not significant even
%with binary quantization and it could be improved by using higher
%level quantizers. We also show that it is always preferable, at least
%from a mutual information point of view, to utilize the real and
%imaginary parts of channel coefficients instead of the amplitudes and
%phases for key generation. Finally, we design both regular and
%irregular Slepian-Wolf LDPC codes to reconcile the information between
%Alice and Bob. We also design a non-binary (4-ary or $GF(4)$)
%Slepian-Wolf LDPC code to meet the need of a higher level quantizer.

\subsection{Notation and outline}

Unless otherwise specified, we use upper case letters, e.g., $X$ to
denote random variables and bold uppercase, e.g., $\bX$ to denote
random vectors; $x$ and $\bx$ are their respective realizations. If
$X$ is a complex random variable, we use $\fRe(X)$ and $\fIm(X)$ to
denote, respectively, the real and imaginary parts of $X$. We use $X
\sim \cCN(m, \sigma^2)$ to denote a complex Gaussian random variable
$X$ with mean $m$, variance $\sigma^2$, and with real and imaginary
parts independent and identically distributed. 

The rest of the paper is organized as follows. In
Sec.~\ref{sec.background} we provide background material on the OFDM
channel model. In Sec.~\ref{sec.sysModelDefs} we define secret key
capacity and introduce the measurement model.  In
Sec.~\ref{sec.capCalcs} we evaluate the secret key capacity for
  various channels of interest, and draw a number of useful lessons
  for designs.  In Sec.~\ref{sec.implement} we describe our
  designs and algorithms. In Sec.~\ref{sec.results} we provide
  numerical results for typical 802.11a parameter settings and secret key capacity. We
  conclude in Sec.~\ref{sec.conclusion}.  Some proofs are deferred to the
  Appendix.

\section{Channel Diversity: An OFDM Example}
\label{sec.background}

In this section, we introduce diversity signaling technique used in OFDM system. We then characterize the channel coefficients which represent the channel diversity and from which we extract our secret keys. 
%Since we extract the keys from sampled channel coefficients, we use
%the term ``sampled channel coefficients'' and ``channel coefficients''
%interchangablly.
The OFDM model we use follows closely that introduced in
\cite{Sayeed_book}.

\subsection{OFDM signaling} 
Let $T$ denote the signaling duration and $W$ denote the two-sided
bandwidth of a wireless link with $M = TW$. An OFDM system
transmits $M$ orthogonal signals.  The transmitted signal $s(t)$ can be represented as
\begin{equation}
s(t) = \sum_{n=0}^{M-1} s_n \phi_n(t), \mbox{	} 0 \leq t \leq T,
\end{equation}   
where the $s_n$ are the information-bearing signal coefficients and
the $\phi_n(t)$ are the orthogonal modulating waveforms or
``tones''.  In OFDM the Fourier basis is used, i.e.,
\begin{equation}
\phi_n(t) = \left\{ \begin{array}{lcl} \frac{1}{\sqrt{T}} e^{j2\pi n \Delta f t}& \mbox{if} & 0 \leq t \leq T\\ 0 & \mbox{else} \end{array} \right.
\end{equation}  
where $\Delta f = \frac{1}{T}$.  The received signal $r(t)$ is $r(t) = h(t) \ast s(t) + w(t)$ where $h(t)$ is the communication channel, assumed to be
time-invariant during the two-way training, $w(t)$ is the
receiver noise, and $\ast$ denotes continuous-time convolution.  We
model $w(t)$ as a complex zero-mean white Gaussian noise process with
autocorrelation function $E[w(t_1) w^{\ast}(t_2)] = \varNoise
\delta(t_1 - t_2)$ where $\delta(\cdot)$ is the Dirac delta. 

We discretize the observation $r(t)$ by projecting it onto the orthogonal basis
functions $\phi_n(t)$ to produce
\begin{eqnarray}
\label{ofdm}
r_n = \int_{- \infty}^{\infty} r(t) \phi^{*}_n(t) dt = H_n s_n + w_n,
\end{eqnarray} 
where 
\begin{equation*}
H_n = \int_{- \infty}^{\infty} \sqrt{T} h(t) \phi^{*}_n(t) dt
\end{equation*}
is the {\em frequency domain} channel coefficient at the $n^{th}$ tone
and the
\begin{equation*}
w_n = \int_{- \infty}^{\infty} w(t) \phi^{*}_n(t) dt
\end{equation*}
are independent zero-mean complex Gaussian random variables of
variance $\varNoise$.

Wireless multipath channels $h(t)$ are well modeled in \cite{Sayeed_book}
as having an echo-type impulse response.  In particular, let
\begin{equation}
h(t) = \sum_{k=1}^{N_p} \beta_k \delta(t -
\tau_k) \label{eq.defEcho}
\end{equation}
where $N_p$ is the total number of propagation paths, and $\tau_k \in
[0, \tau_{\max}]$, $\tau_{\max}$ and $\beta_k$ are the delay, the delay spread and the
complex channel gain associated with the $k^{th}$ path.  Since
$\tau_k$ is typically much longer than the speed of light divided by
the carrier, each $\beta_k$ is well modeled as having uniform random phase. Also, 
since the scaterring objects are distinct, $\beta_k$ are well modeled 
as independent random variables. We incorporate an exponential power-delay 
profile where the variance of $\beta_k$ decays with $\tau_k$. 

The frequency domain channel coefficients $H_n$, $0 \leq n \leq M-1$ are 
\begin{align}
H_n & = \sum_{k=1}^{N_p} \beta_k e^{-j2 \pi \frac{n}{T} \tau_k}\label{freq_coeff}\\ 
& \approx \frac{1}{\sqrt{M}} \sum_{l=0}^{M-1}
h_{\ell} e^{-j 2 \pi \frac{n\ell}{M}},\label{eq.approxCoeff}
\end{align}
where in~(\ref{eq.approxCoeff}) we approximate $H_n$ by quantizing the
$\tau_k$ into $M$ delay bins and aggregating the effect of the
$\beta_k$ terms into the $h_{\ell}$.  Each bin is of length $\tau_{\rm
  bin} = 1/W$ and
\begin{align}
h_{\ell} & = \sqrt{M} \sum_{k=1}^{N_p} \beta_k {\rm sinc} \left[ W
  \left(\frac{\ell}{W} - \tau_k\right) \right] \label{eq.sinc}\\ &
\approx \sqrt{M} \sum_{k : \frac{\ell - 0.5}{W} < \tau_k \leq \frac{\ell +
    0.5}{W}} \beta_k. \label{eq.noSinc}
\end{align}
The variable $h_{\ell}$ is the \textit{sampled} or \textit{time
  domain} channel coefficient associated with the ${\ell}^{th}$ {\em
  resolvable} delay bin.  If there are many $\beta_k$ associated with
each bin, as is the case for rich multipath, the $h_{\ell}$ are well
approximated as zero-mean complex Gaussian random variables; a further
approximation justified by the central limit theorem. 

In OFDM channel, $\tau_{\max} \leq T$, thus only the first few
delay bins have physical paths contributing to them, similarly only the first few $h_{\ell}$
will be significant.  Say the first $L \leq M$ sampled channel
coefficients are significant, then we further simplify our
approximation of $H_n$ as
\begin{align}
H_n & \approx \frac{1}{\sqrt{M}} \sum_{\ell= 0}^{L-1} h_{\ell} e^{-j 2
  \pi \frac{n\ell}{M}}. \label{eq.nonZeroCoeff}
\end{align} 
where we have neglected the effect of the tails of the sinc waveforms
in~(\ref{eq.sinc}).

The $h_{\ell}$ is well modeled as having uniform phase (as remarked
following~(\ref{eq.defEcho})) and having a complex Gaussian
distribution (as remarked following~(\ref{eq.noSinc})).  Since the
paths aggregated into distinct $h_{\ell}$ are typically reflections
from distinct scatters, the $L$ non-zero $h_{\ell}$ are also often
well modeled as being statistically independent.  However, the
$h_{\ell}$ are not identically distributed; the variance is roughly
inversely proportional to $\tau_k$ as the result of the exponential power-delay
profile on $\beta_k$. On the other hand, $H_n$ exhibits
Gaussian characteristic under rich multipath with variance $\sigma_H^2$. Following from (\ref{freq_coeff}), we have 
\begin{align*}
\sigma_H^2 &= E\left[|H_n|^2\right] = \sum_{k = 1}^{N_p} E\left[ |\beta_k|^2 \right],
\end{align*}
which does not depend on $n$. Hence, while the $H_n$ are not independent, they have the same marginal distribution. 

We note that if there are only a few transmission paths, the assumption that channel coefficients are Gaussian distributed no longer holds. However, we are using Gaussian model as an example to illustrate the importance of exploiting channel diversity, which is actually not limited to Gaussian case.

\subsection{Signal-to-noise ratio}
As mentioned above, when multipath is rich, i.e., $N_p$ is large, the
$H_n$ can be well modeled as $\cCN(0, \sigma_H^2)$.
We define the {\em per-tone} SNR as
\begin{eqnarray}
\label{freq_snr}
SNR_f = \frac{E[H_n^2]}{E[w_n^2]} = \frac{\sigma_H^2}{\varNoise}.
\end{eqnarray}
%Because the DFT is a unitary transform, Parseval's theorem states that
%\begin{eqnarray}
%\sum_{n=0}^{M-1} |H_n|^2 = M \sum_{k=1}^{N_p} |\beta_k|^2 \approx
%\sum_{\ell = 0}^{M-1} |h_{\ell}|^2 \approx \sum_{\ell = 0}^{L-1}
%|h_{\ell}|^2,
%\end{eqnarray}
%where the approximation follows from the approximations above.
%Also as discussed above, $h_{\ell}$ is well modeled as $\cCN(0,
%\sigma_h^2(\ell))$.  We then have the following relationship between
%the variance of frequency domain channel coefficients and sampled
%coefficients:
%\begin{eqnarray}
%\sum_{\ell = 0}^{L-1} \sigma_h^2(\ell) = \sum_{\ell = 0}^{L-1}
%E[h_{\ell}^2] \approx \sum_{n=0}^{M-1} E[H_n^2] = M \sigma_H^2,
%\end{eqnarray}
Also as discussed above, $h_{\ell}$ is well modeled as $\cCN(0,
\sigma_h^2(\ell))$. We can thus define the {\em time-domain} SNR as
\begin{equation}
SNR_{\tau}(\ell) = \frac{\sigma_h^{2}(\ell)}{\varNoise}.
\end{equation}
It can be shown that we have the relation
\begin{eqnarray}
\sum_{\ell = 0}^{L-1} SNR_{\tau}(\ell) \approx M \cdot SNR_f.
\end{eqnarray}
If the sampled channel coefficients have equal variance, the
relationship simplifies to
\begin{eqnarray}
SNR_{\tau} \approx \frac{M}{L} SNR_f.
\end{eqnarray}

\section{Secret key systems: Definitions and measurement model}
\label{sec.sysModelDefs}

%%%%%%%%%%%%%%%%%%%%%%%%%%%%%%

In this section we introduce the measurement model, secret key
generation system, and study the secret key capacity.

\subsection{System model}
\label{sec.sysModel}

Secret key generation system has been studied by many authors. In
particular, the authors in~\cite{maurer93, Ahlswede} study the fundamental limits on the
achievable secret key rates. We state their results for reference in
the context of our application.

\begin{definition}
\label{def.secretKeySys}
A length-$N$ \textit{secret key generation system} over alphabets
$\cX_A, \cX_B, \cK, \cS$ is a triplet of functions ($f$, $g_A$, $g_B$):
\begin{eqnarray}
f : \cX_A^{N} \!&\rightarrow&\! \cK^{NR}, \label{eq.keyGen}\\ g_A :
\cX_A^{N} \!&\rightarrow&\! \cS^{m}, \label{eq.syndromeGen} \\ g_B :
\cX_B^{N} \times \cS^{m} \!&\rightarrow&\!
\cX_A^{N}. \label{eq.keyRec}
\end{eqnarray}
\end{definition}

We interpret this definition in the context of the system operation
described in Sec.~\ref{sec.sysOverview}. The function $f$ maps Alice's source
of randomness into the secret key.  The function $g_A$ defines the
public message Alice sends to Bob.  The function $g_B$ is Bob's
decoding function that maps his observation and the public message
into his estimate of Alice's observation.  If Bob's estimate is correct
applying $f(\cdot)$ to it will recover the key.

%We let $W = 1-R'$ such that the key extracting rate can be viewed as
%the fraction of observation that is not revealed to the public.

Given a source of randomness $p_{X_A^N, X_B^N}(x_A^N, x_B^N)$, where
$x_A^N \in {\cal X}_A^N$ and $x_B^N \in {\cal X}_B^N$, the secret key capacity
is the supremum of achievable secret key rates.  An achievable secret key
rate is defined as follows.
\begin{definition}
A secret key rate $R$ is \textit{achievable} if for any $\epsilon > 0$
and $N$ sufficiently large, we have:
\begin{align}
N R \log |\cK| - H(f(X_A^N)) & \leq \epsilon, \\
\Pr\left[f(X_A^N) \neq f(g_B(X_B^N, g_A(X_A^N)))\right] & < \epsilon, \\
\frac{1}{N} I(f(X_A^N); g_A(X_A^N)) & \leq \epsilon. \label{eq.secrecyCond}
\end{align}
\end{definition}
The first inequality implies that the secret key is nearly uniformly
distributed.  The second inequality upper bounds the probability of
error in key recovery.  The final inequality is the secrecy guarantee,
i.e., that the public message tells you little about the key.

The above definitions are often stated for a setting in which an
attacker (Eve) has access to a correlated measurement of the source
{\em as well as} the public message.  We do not include this
possibility in the definitions as stated herein due to the source of
randomness we study.  We will characterize secret key capacity for an
OFDM system where the correlated observations $X_A^N$ and $X_B^N$ are
functions of the underlying channel law (the $H_n$ or the $h_{\ell}$
of Sec.~\ref{sec.background}).  In a rich scattering environment the
channel law between two users changes utterly if either moves more
than a few wavelengths (a few centimeters for an OFDM system).
Therefore an eavesdropper would have to be positioned extremely close
to either Alice or Bob to get useful channel observations.  This is
 one of the inherent strengths of this source of randomness -- it is
difficult to eavesdrop.  And for this reason we ignore the possibility
of eavesdropping throughout the rest of the paper.  (In contrast, the
public message is easy to intercept, and so we must assume Eve has
knowledge of that message.)

In~\cite{Ahlswede} the following theorem is shown
\begin{theorem}
For a discrete memoryless source $p_{X_A^N, X_B^N}(x_A^N, x_B^N)$ the
secret key capacity is
\begin{equation} 
C = \lim_{N \rightarrow \infty} \frac{1}{N} I(X^N_A; X^N_B), \label{eq.defCap}
\end{equation}
assuming the limit exists.
\end{theorem}

%Furthermore, we are not concerned about how keys are extracted from
%observation, i.e., the function $f$. We therefore ignore this function
%throughout the paper and concentrate on the reconciliation of Alice
%and Bob's observations. It is natural to assume that Bob quantizes his
%observation using the same function $q$ as Alice does. However, this
%is not required, as discussed in later sections. Whether or not Bob
%quantizes his observation is analogous to hard or soft decoding.

\subsection{Measurement model}
The sources of randomness we work with in the paper are noisy
measurements of the channel coefficients.  Alice and Bob each sends an
identical and known sounding signal $s(t)$ to the other.  For
simplicity we assume each signal coefficient $s_n = 1$.  (Equal-power
sounding is known not always to be the best choice, see~\cite{Chou1}.)
%Most of the off-the-shelf Wi-Fi transceivers are time division duplex
%(TDD) transceivers, which makes it impossible for both Alice and Bob
%to sound the channel and measure the channel at the same time.  Hence,
%we have Alice and Bob to transmit their sounding signals in turn, and
We assume that the channel remains static during this two-way training.
The period in which a wireless channel is roughly static is termed the
{\em coherence period}.  Thus, this two-way training is assumed to
occur within a single coherence period.

Under this channel assumption we model Alice and Bob's measurements
as
 \begin{align}
 \label{observ_model}
 H_{A, n} &= H_n + w_{A, n} \nonumber \\
 H_{B, n} &= H_n + w_{B, n},
 \end{align}
respectively, where $w_{A, n}$, $w_{B, n} \sim \cCN(0, \varNoise)$ are
independent sources of noise. We notice that the phase offset caused by the local oscillators may add extra noise to the measurement \cite{WuBarness, Liang1, ProxiMate}. We defer the discussion of phase offset to the end of this section.

The frequency domain correlation
coefficient between Alice and Bob's observation at $n^{th}$ tone can be shown to be :
\begin{align}
\rho_f &= \frac{SNR_f}{1+SNR_f}.
\end{align}
Note that the correlation $\rho_f$ between $H_{A, n}$ and $H_{B, n}$
is equal to the correlation between $\fRe(H_{A, n})$ and $\fRe(H_{B,
  n})$ and is also equal to that between $\fIm(H_{A, n})$ and
$\fIm(H_{B, n})$. We can also consider the time domain observation as:
\begin{align}
\label{observ_model_time}
 h_{A, \ell} &= \fRe(h_{A, \ell}) + j \fIm(h_{A, \ell}) = h_{\ell} + n_{A, \ell} \nonumber \\
 h_{B, \ell} &= \fRe(h_{B, \ell}) + j \fIm(h_{B, \ell}) = h_{\ell} + n_{B, \ell},
\end{align}
where $h_{\ell} \sim \cCN(0, \sigma_h^2(\ell))$ is the sampled channel
coefficient and $n_{B, \ell}$, $n_{A, \ell} \sim \cCN(0, \varNoise)$
are the noises. Similar to the correlation in frequency domain, the
correlation coefficient in $\ell^{th}$ sampled channel coefficient is
given as:
\begin{align}
\label{corr_coeff_time}
\rho_{\tau}(\ell) &=
\frac{SNR_{\tau}(\ell)}{1+SNR_{\tau}(\ell)}.
\end{align}
Note again that the correlation coefficient $\rho_{\tau}(\ell)$
between $h_{A, \ell}$ and $h_{B, \ell}$ is equal to the correlation
coefficient between $\fRe(h_{A, \ell})$ and $\fRe(h_{B, \ell})$ or
equivalently equal to that between $\fIm(h_{A, \ell})$ and $\fIm(h_{B,
  \ell})$.

To get to the long block-lengths possibly required to approach secret key
capacity, we repeat this two-way channel sounding across multiple
channel coherence periods.  The channel is assumed to be independently
and identically distributed across coherence periods.  Say that
within each period Alice and Bob generate channel observations $h_{A,
  \ell}$ and $h_{B, \ell}$, respectively, for $\ell = 1, 2, \ldots,
L$.  Further, say they do this for $n$ coherence periods yielding
measurements $h_{A, \ell}[i]$ and $h_{B, \ell}[i] $ for $i = 1, 2,
\ldots, n$.  They stack their observations into the length-$N$ real
vectors, where $N = 2nM$ as follows:
\begin{equation}
X_A^N = \left[ \begin{array}{c} \fRe(h_{A,1}[1])
    \\ \fIm(h_{A,1}[1])\\ \vdots \\ \fRe(h_{A,L}[1])
    \\ \fIm(h_{A,L}[1])\\ \fRe(h_{A,1}[2])\\ \vdots
    \\ \fIm(h_{A,L}[n]) \\ 0 \\ \vdots \\ 0
\end{array} \right], \hspace{1em} 
X_B^N = \left[ \begin{array}{c} \fRe(h_{B,1}[1])
    \\ \fIm(h_{B,1}[1])\\ \vdots \\ \fRe(h_{B,L}[1])
    \\ \fIm(h_{B,L}[1])\\ \fRe(h_{B,1}[2])\\ \vdots
    \\ \fIm(h_{B,L}[n]) \\ 0 \\ \vdots \\ 0
\end{array} \right], \label{eq.stackObs}
\end{equation}
where the padding is with $2nM - 2nL = 2n(M-L)$ zeros.  These are the
degrees of freedom lost due to the fact that the last $M-L$
coefficients in each block are zero, cf.~(\ref{eq.nonZeroCoeff}).
Were the approximation that the last $M-L$ coefficients in each block
were zero to be exact, then due to the i.i.d.\ assumption across
coherence blocks, the limit in~(\ref{eq.defCap}) would exist and would
evaluate to
\begin{equation}
C = \lim_{N \rightarrow \infty} \frac{1}{N} I(X_A^N; X_B^N) =
\frac{1}{2M} I(h_A^L; h_B^L) \label{eq.capLimit}
\end{equation}
where $h_A^L$ and $h_B^L$ are, respectively, the length-$L$ complex
vectors of observations made by Alice and Bob.  While the definitions
provided in Sec.~\ref{sec.sysModel} are for finite alphabets, the extension to
continuous alphabets follows from standard limiting arguments.

\subsection{Phase Offset}
\label{sec.PhaseNoise}

We have thus far implicitly assumed the perfect synchronization between Alice and Bob. In practice, however, Alice and Bob measured channel parameters may be effected by the phase offset caused by the local oscillators of both transmitters. Since phase synchronization is not perfect, there is a phase offset during each channel sounding. Furthermore, the frequency generated by local oscillators continuously fluctuates (or drifts) around its center frequency, causing a time dependent phase drift. There are many existing techniques developed to mitigate the effect of such phase offset (see \cite{WuBarness}, \cite{Liang1} and the references therein). 

Since the signal duration is very small in a channel training, we assume that during each channel sounding phase offset caused by oscillator frequency drift is negligible, i.e., the phase offset is time invariant. However we do not assume it is negligible across channel trainings, i.e., between coherence intervals. Denote the phase offset caused by Alice and Bob's local oscillators as $\theta_A$ and $\theta_B$ respectively. The offsets can be incorporated into Alice's and Bob's measurements as $h_A^L e^{j\theta_A}$ and $h_B^L e^{j\theta_B}$ cf.~(\ref{eq.noSinc}). Since phase offset is differential, without the loss of generality, we can incorporate the error into Bob's measurement and write $h_A^L$ and $h_B^L e^{j\theta}$ with $\theta = \theta_B - \theta_A$. Then the unnormalized secret key capacity in (\ref{eq.capLimit}) becomes $I(h_A^L; h_B^L e^{j \theta})$. 

We show that by exploiting the channel diversity, one can mitigate the effect on the secret key capacity caused by phase offset. First, note that 
\begin{align*}
I(h_B^L e^{j \theta}; h_A^L, e^{j \theta}) &= I(h_B^L e^{j \theta}; h_A^L)  + I(h_B^L e^{j \theta}; e^{j \theta} | h_A^L) \\
&= I(h_B^L e^{j \theta}; e^{j \theta}) + I(h_B^L e^{j \theta}; h_A^L | e^{j \theta}).
\end{align*}
Then we can write:
\begin{align}
\label{eq.haha}
I(h_B^L e^{j \theta}; h_A^L) =& I(h_B^L e^{j \theta}; h_A^L | e^{j \theta}) + I(h_B^L e^{j \theta}; e^{j \theta}) \nonumber \\
&- I(h_B^L e^{j \theta}; e^{j \theta} | h_A^L) \nonumber \\
=& I(h_B^L; h_A^L) \! + \! I(h_B^L e^{j \theta}; e^{j \theta}) - I(h_B^L e^{j \theta}; e^{j \theta} | h_A^L) \nonumber \\
\stackrel{(a)}{\geq}&  I(h_B^L; h_A^L) - I(h_B^L e^{j \theta}; e^{j \theta} | h_A^L). \nonumber \\
=& I(h_B^L; h_A^L) - I(e^{j \theta}; h_A^L, h_B^L e^{j \theta}). 
\end{align}
Inequality $(a)$ is equality ($I(h_B^L e^{j \theta}; e^{j \theta}) = 0$) if $h_B^L$ is circularly symmetric complex Gaussian vector since $h_B^L$ and $h_B^L e^{j \theta}$ have the same distribution. The last equality follows because $e^{j \theta}$ is independent of $h_A^L$. The second term on the right hand side of (\ref{eq.haha}), $I(e^{j \theta}; h_A^L, h_B^L e^{j \theta})$, is the secret key capacity loss caused by the phase offset and it is the decrease in uncertainty about the unknown offset $e^{j \theta}$ given the knowledge of $h_A^L$ and $h_B^L e^{j \theta}$ as measured in bits. Note that because $h_B^L$ and $h_B^L e^{j \theta}$ have the same distribution nothing can be learned about $\theta$ by observing $h_B^L e^{j \theta}$ only. However in combining with the knowledge of $h_A^L$ one can better estimate $\theta$ because $\angle h_B^L e^{j \theta} = \angle h_B^L + \theta$ and $\angle h_B^L$, $\angle h_A^L$ are dependent random variables. Thus we get $L$ independent looks at $\theta$ with additive noise (since $h_A^L$ and $h_B^L$ each has $L$ independent entries). Note that because of the additive noise, it is impossible to estimate $\theta$ with infinite precision.

Since $\theta$ is a scalar the loss term does not scale linearly in $L$. By the Cram\'{e}r-Rao bound we know that variance of the estimate of $\theta$ can drop at most as $\frac{1}{L}$ which means for general distribution the loss should scale as $\log L$. In other words, as $L$ becomes larger while one can potentially get better estimate of $\theta$ from channel observations, the loss is scaling more slowly than the gain from the first term on the right hand side of (\ref{eq.haha}), which scales linearly in $L$. This supports our claim that channel diversity should be exploited, both as a way to boost secret key capacity and to mitigate the phase offset.

Later in Sec.~\ref{sec.implement} we will show how the LDPC design can be adapted to perform the estimation of phase offset.

\section{Secret key capacity calculations}
\label{sec.capCalcs}

We are now in position to evaluate the secret key capacity for various
channels of interest.  In Sec.~\ref{sec.capCSI} we first do this for
the general OFDM model of time-domain channel coefficients.  Then, to
ease analysis, we focus on an idealized model wherein all sampled
channel coefficients have the same variance.  This simplification
allows us to draw a number of general lessons on secret key generation
for OFDM channels.  In Sec.~\ref{sec.RSSI} we quantify the (quite
large) reduction in secret key rate when only received signal strength
indicator (RSSI) information is available, as opposed to full CSI.  In
Sec.~\ref{sec.ampPhase} we discuss generating keys separately from the
amplitude and phase of the CSI, as opposed to the real and imaginary
parts.  %Because of correlation between the amplitude measurement of
%Alice and the phase of Bob, and vice-versa, there is a loss from doing
%this.  Finally, in Sec.~\ref{sec.quant} we discuss the impact of
%quantization.

%%%%%%%%%%%%%%%%%%%%%%%%%%%%%%%%%%%%%%
\subsection{Secret key capacity based on CSI}
\label{sec.capCSI}

We now evaluate~(\ref{eq.capLimit}) in terms of the SNR of the
channel.  Due to the fact that $h_A^L$ and $h_B^L$ are jointly complex
Gaussian and i.i.d.\ in time, we have
\begin{eqnarray}
\label{secret_key_capacity}
C \approx  -\frac{1}{2M} \sum_{\ell = 0}^{L-1} \log \bigg [ 1- \bigg
  (\frac{SNR_{\tau}(\ell)}{1+SNR_{\tau}(\ell)} \bigg )^2 \bigg ],
\end{eqnarray}
where the approximation follows from the fact that the last $M-L$ sampled coefficients are 
approximately zero. If the sampled coefficients have equal variance, the capacity
simplifies to
\begin{eqnarray}
\label{secret_key_capacity_ideal}
C \approx - \frac{L}{2M} \log \bigg [ 1- \bigg
  (\frac{SNR_{\tau}}{1+SNR_{\tau}} \bigg )^2 \bigg ].
\end{eqnarray}

Note that the correlation coefficient in time relates to that in
frequency as:
\begin{equation}
\label{rho_tau}
\rho_{\tau} %&=& \rho_{\tau}(\ell) 
\approx \frac{M \cdot SNR_f}{L + M \cdot SNR_f} 
= \frac{M \rho_f}{L + (M-L) \rho_f}.
\end{equation}

%We comment that if we have access to the statistics of frequency
%domain coefficients, the secret key capacity between Alice and Bob can
%be equivalently calculated from frequency domain channel
%coefficients. To do this, we perform eigenvalue decomposition (EVD) on
%the correlation matrix of frequency domain coefficients thus
%completely decorrelate those frequency domain coefficients into
%independent entities and the mutual information is the sum of the
%mutual information of these entities. The secret key capacity is then
%the total mutual information normalized by $2M$. However, we may not
%be able to do this in the actual operation of the system because the
%training between Alice and Bob may only occur few times before they
%have enough statistics of frequency domain coefficients unless it is
%an available prior.

In the remainder of this section we focus on an idealized model
wherein all sampled channel coefficients have the same variance. We
let $h_{A, \ell}$, $h_{B, \ell} \sim \cCN(0, \sigma^2)$ where
$\sigma^2 = \sigma_h^2 + \varNoise$. Note that $\fRe (h_{A, \ell})$,
$\fRe (h_{B, \ell})$ have correlation coefficient $\rho_{\tau}$
defined in~(\ref{rho_tau}) and $\fIm (h_{A, j})$, $\fIm (h_{A, j})$
also have the same correlation coefficient $\rho_{\tau}$. The secret
key capacity between Alice and Bob now reduces
to~(\ref{secret_key_capacity_ideal}).

%%%%%%%%%%%%%%%%%%%%%%%%%%%%%%%

\subsection{Secret key generation based on measurements of RSSI}
\label{sec.RSSI}

In this section we compare the secret key capacity given sampled
channel coefficients~(\ref{secret_key_capacity_ideal}) to the secret
key capacity if only receiver signal strength indicator (RSSI) values
are available.  Since RSSI summarizes the true vector of channel state
information, there will clearly be a reduction in secret key capacity
if only RSSI values are made available.  In fact the reduction is dramatic.  
From a technological point of view, most
off-the-shelf wireless transceivers make only RSSI values available to
the upper layers, not the channel state information.  This section
demonstrates that making full CSI available would greatly help the
ability to generate secret keys.

To calculate the secret key capacity based on RSSI values, let $R_A$
and $R_B$ denote the RSSI values received by Alice and Bob. In an OFDM
system, the RSSI takes the form \cite{80216}:
\begin{align*}
R_A & = \! \sum_{\ell=0}^{L-1} |h_{A, \ell}|^2 = \!
\sum_{\ell=0}^{L-1} \! \left[ |\fRe(h_{A, \ell})|^2 \! \! + \!
  |\fIm(h_{A, \ell})|^2 \right] = \! \! \sum_{\ell=0}^{2L-1} \!
X_{A,\ell}^2,
%R_B & = \! \sum_{\ell=0}^{L-1} |h_{B, \ell}|^2 = \! \sum_{\ell=0}^{L-1}
%\! \left[|\fRe(h_{B, \ell})|^2 \!\! + \! |\fIm(h_{B, \ell})|^2\right] =
%\!\! \sum_{\ell=0}^{2L-1} \! X_{B,\ell}^2,
\end{align*}
where 
\begin{align*}
X_{A,\ell} & = \left\{ \begin{array}{lcl} |\fRe(h_{A, \ell})|^2 &
  \mbox{if} & 0 \leq \ell \leq L-1 \\ |\fIm(h_{A, \ell-L})|^2 &
  \mbox{if} & L \leq \ell \leq 2L-1 \end{array} \right.
%X_{B,\ell} & = \left\{ \begin{array}{lcl} |\fRe(h_{B, \ell})|^2 &
 % \mbox{if} & 0 \leq \ell \leq L-1 \\ |\fIm(h_{B, \ell-L})|^2 &
  %\mbox{if} & L \leq \ell \leq 2L-1 \end{array} \right.
\end{align*}
The quantities $R_B$ and $X_{B, \ell}$ are defined similarly. Further, $X_{A,\ell}$ and $X_{B,\ell}$ are $\cN(0, \frac{\sigma^2}{2})$
Gaussian random variables with $\rho_{\tau} = E[X_{A, \ell} X_{B,
    \ell}]$ for all $\ell$. 

Both $R_A$ and $R_B$ are non-standard chi-square distributed random
variables with $2L$ degree of freedom. %However, since $I(R_A; R_B)$
%equals to $I(\alpha R_A; \alpha R_B)$ for any non-zero constant
%$\alpha$, we can normalize $X_{A,\ell}$ and $X_{B,\ell}$ by
%$\sigma/\sqrt{2}$ to get, respectively, $R_A$ and $R_B$ which are
%standard chi-square distributed random variables with $2L$ degree of freedom. 
The joint probability density function of a pair of
chi-square random variables is given in Theorem 2.1 in \cite{Joarder_1} and we use it to numerically calculate the mutual information between $R_A$ and $R_B$, denoting the secret key capacity calculated as:
\begin{eqnarray}
\label{secret_key_capacity_rssi}
C_{R} = \frac{1}{2M} I(R_A; R_B).
\end{eqnarray}
%\begin{theorem} 
 %Two random variables $R_A$ and $R_B$ are said to have a correlated
 %bivariate chi-square distribution, each with $m$ degrees of freedom,
 %if the joint probability density function is given in \cite{Joarder_1}.
%\begin{align}
%f_{R_A R_B}(r_A, r_B)
%&=\frac{2^{-(m+1)}(r_Ar_B)^{(m-2)/2}e^{\frac{-(r_A+
 %     u2)}{2(1-\rho_{\tau}^2)}}}{\sqrt{\pi} \Gamma(m/2)
  %(1-\rho_{\tau}^2)^{m/2}} \nonumber \\ &\cdot~ \sum_{k=0}^\infty
%[1+(-1)^k] \bigg (\frac{\rho_{\tau}
 % \sqrt{r_Ar_B}}{1-\rho_{\tau}^2}\bigg)^k \nonumber \\ &\cdot~
%\frac{\Gamma (\frac {k+1}{2} )}{k! \Gamma (\frac{k+m}{2})},
%\end{align}
%where $m > 2$, $-1 < \rho_{\tau} < 1$, and $\Gamma(\cdot)$ is the
%Gamma function \cite{Joarder_1}.
%\end{theorem}

%The mutual information between a pair of chi-square random variables
%does not have a closed form expression so we calculate it
%numerically. Let $f_{R_A}(r_A)$ and $f_{R_B}(r_B)$ denote the
%respective marginal probability density function of $R_A$ and
%$R_B$. The mutual information between $R_A$ and $R_B$ is:
%\begin{align}
%\label{mutual_rssi}
%I(R_A; R_B) &= \int \int f_{R_A R_B}(r_A, r_B) \nonumber \\ 
%&\cdot~ \log \frac{f_{R_A R_B}(r_A, r_B)}{f_{R_A}(r_A) f_{R_B}(r_B)} dr_A dr_B,
%\end{align} 
%and the secret key capacity is:
%\begin{eqnarray}
%\label{secret_key_capacity_rssi}
%C_{R} = \frac{1}{2M} I(R_A; R_B).
%\end{eqnarray}

When $L$ is large, $R_{A}$ and $R_{B}$ can be well approximated as
Gaussian random variables $\cN(2L, 4L)$ due to the central limit
theorem. The mean and variance for $R_{A}$ and $R_{B}$ can be
calculated as $E[R_A] = \sum_{\ell = 0}^{2L-1} E[X_{A, \ell}^2] $ and
$var(R_A) = \sum_{\ell = 0}^{2L-1} var(X_{A, \ell}^2)$ using the
identities $E[X_{A,\ell}^2] = 1$, $E[X_{A,\ell}^4] = 3$, and
$E[(X_{A,\ell}^2 - E[X_{A,\ell}^2])^2] = 2$, which follow from the
variance normization.  The correlation coefficient between $R_{A}$ and
$R_{B}$ is
\begin{align}
\label{rho}
\rho_{R} &= \frac{E[( R_{A} - 2L) ( R_{B} - 2L )]}{4L} \nonumber \\
&= \frac{E[(X_{A,\ell}^2 - 1) ( X_{B,\ell}^2 - 1)]}{2},
\end{align}
where the joint moment generating function of $X_{A,\ell}$ and $X_{B,\ell}$
is:
\begin{eqnarray}
M(s_1, s_2) = E[e^{s_1X_{A,\ell} + s_2 X_{B,\ell}}] =
e^{[\frac{1}{2}(s_1^2 + 2\rho_{\tau} s_1s_2 + s_2^2)]}.
\end{eqnarray}
We calculate the joint moment $E[X_{A,\ell}^2 X_{B,\ell}^2]$ by taking
second order partial derivatives of $M(s_1, s_2)$ respect to $s_1$ and
$s_2$ and evaluate the result at $s_1 = 0$, $s_2 = 0$. Equation
(\ref{rho}) can be reduced to:
\begin{eqnarray}
\rho_{R} = \rho_{\tau}^2,
\end{eqnarray}  
and the secret key capacity based on RSSI under Gaussian approximation
is then:
\begin{eqnarray}
\label{secret_key_capacity_rssi_gaussian}
C_{R} = \frac{1}{4M} \log \bigg (\frac{1}{1-\rho_{\tau}^4} \bigg ).
\end{eqnarray}

Observe from (\ref{secret_key_capacity_rssi_gaussian}) that the secret
key capacity does not depend on $L$. In other words, at a given
$SNR_{\tau}$, while the capacity between coefficients increases
linearly with $L$ as shown in equation
(\ref{secret_key_capacity_ideal}), the capacity between RSSI stays
the same. This is because there is only one single RSSI value
regardless the number of observations. In Fig.~\ref{fig.mutual_info},
we compare the capacity obtained from channel coefficients and from
RSSI for $L = 2, 5$ and $10$ with $M = 10$. The secret key capacity
between the channel coefficients is calculated using
(\ref{secret_key_capacity_ideal}) and that between RSSI is calculated
both using numerical (\ref{secret_key_capacity_rssi}) and Gaussian
approximation (\ref{secret_key_capacity_rssi_gaussian}).
\begin{figure}[!t] 
  \centerline{\epsfig{figure=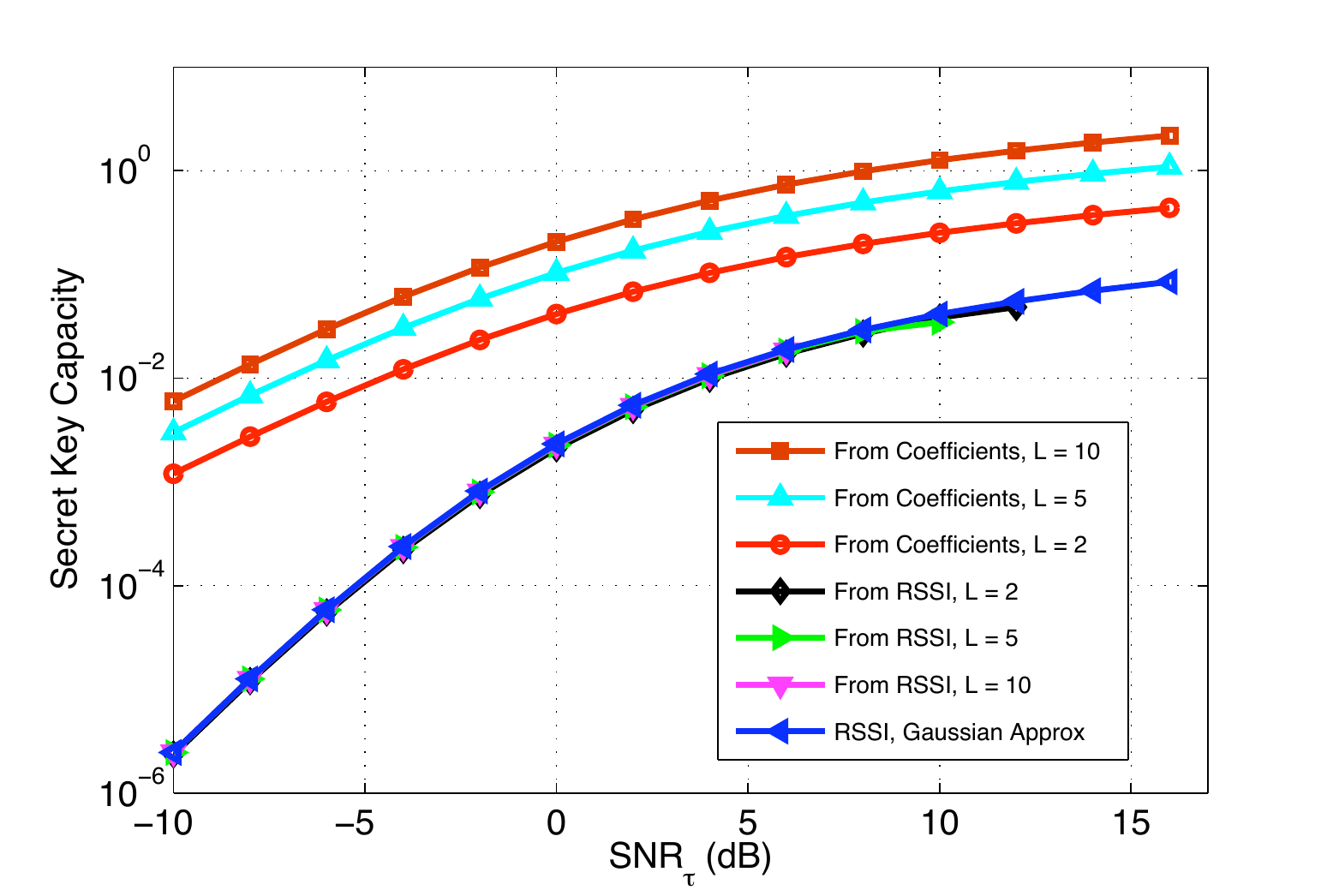,width=9cm}}
  \caption{Secret key capacity when L = 2, 5, and 10. M = 10}
  \label{fig.mutual_info}
\end{figure}
We first note that the secret key capacity obtained from the channel
coefficients increases with $L$, whereas that based on RSSI stays
constant. We also note that the Gaussian approximation is quite accurate, even
when $L$ is rather small.

%%%%%%%%%%%%%%%%%%%%%%%%%%%%%%%%%
\subsection{Representing complex channel coefficients by their real-and-imaginary parts or by their magnitude-and-phase}
\label{sec.ampPhase}

Recall that the secret key capacity~(\ref{eq.capLimit}) is the mutual
information between the sampled channel coefficients observed by Alice
and Bob.  In this section we show this capacity is at least as large
as the sum of the mutual informations between the magnitudes of the
channel coefficients and that between the phases of the channel
coefficients.  This is because while marginally the channel
coefficients observed by Alice and Bob are circularly symmetric (and
thus their magnitude and phase are independent), the correlation
between Alice and Bob's channel coefficients means there is dependence
between Alice's phase and Bob's magnitude and vice-versa.  Thus, the
secret key should not be generated by treating phase and magnitude
separately.  On the other hand, the real parts of Alice and Bob's
coefficients are independent of the imaginary parts.  Thus, without
loss of capacity the key can be generated treating the real and
imaginary parts of each pair of observations as separate pieces of
independent randomness.  This is the reason behind our choice of
definition of $X_A^N$ and $X_B^N$ in~(\ref{eq.stackObs}).

This idea is encapsulated in the following theorem.  For simplicity
(and because the sampled channel coefficients are independent) we focus
on a single pair of observations, $h_A$ and $h_B$.
\begin{theorem} \label{thm.magPhase}
If $h_A$, and $h_B \sim \cCN(0, \sigma^2)$ are jointly complex
Gaussian random variables, we have:
\begin{align}
I(h_A; h_B) &= I(\fRe(h_A); \fRe(h_B)) + I(\fIm(h_A); \fIm(h_B))
\nonumber \\ &\geq I(|h_A|; |h_B|) + I(\Phi_{h_A};
\Phi_{h_B}). \nonumber
\end{align}
\end{theorem}

\begin{proof}
See Appendix~\ref{app.magPhaseProof}.
\end{proof}

In Figure~\ref{fig.mi_comp} we illustrate this result for a range of
SNR.  We plot the capacity $I(h_A; h_B)$, $I(|h_A|; |h_B|) +
I(\Phi_{h_A}; \Phi_{h_B})$, as well as the two terms of the latter,
$I(\Phi_{h_A}; \Phi_{h_B})$ and $I(|h_A|; |h_B|)$. The gap to capacity
is evident at all SNR. It is also worthwhile to note that most of the
information is in the phase information, $I(|h_A|; |h_B|)$ is much
smaller.  This is another illustration of the lesson of
Sec.~\ref{sec.RSSI} as the magnitude information is the RSSI of this
example.  We reiterate that the reason for the gap to capacity is that
the pairs $(|h_A|, |h_B|)$ is not statistically independent of
$(\Phi_{h_A}, \Phi_{h_B})$ due to the correlation between real and
imaginary parts of $h_A$ and $h_B$.

\begin{figure}[!t] 
  \centerline{\epsfig{figure=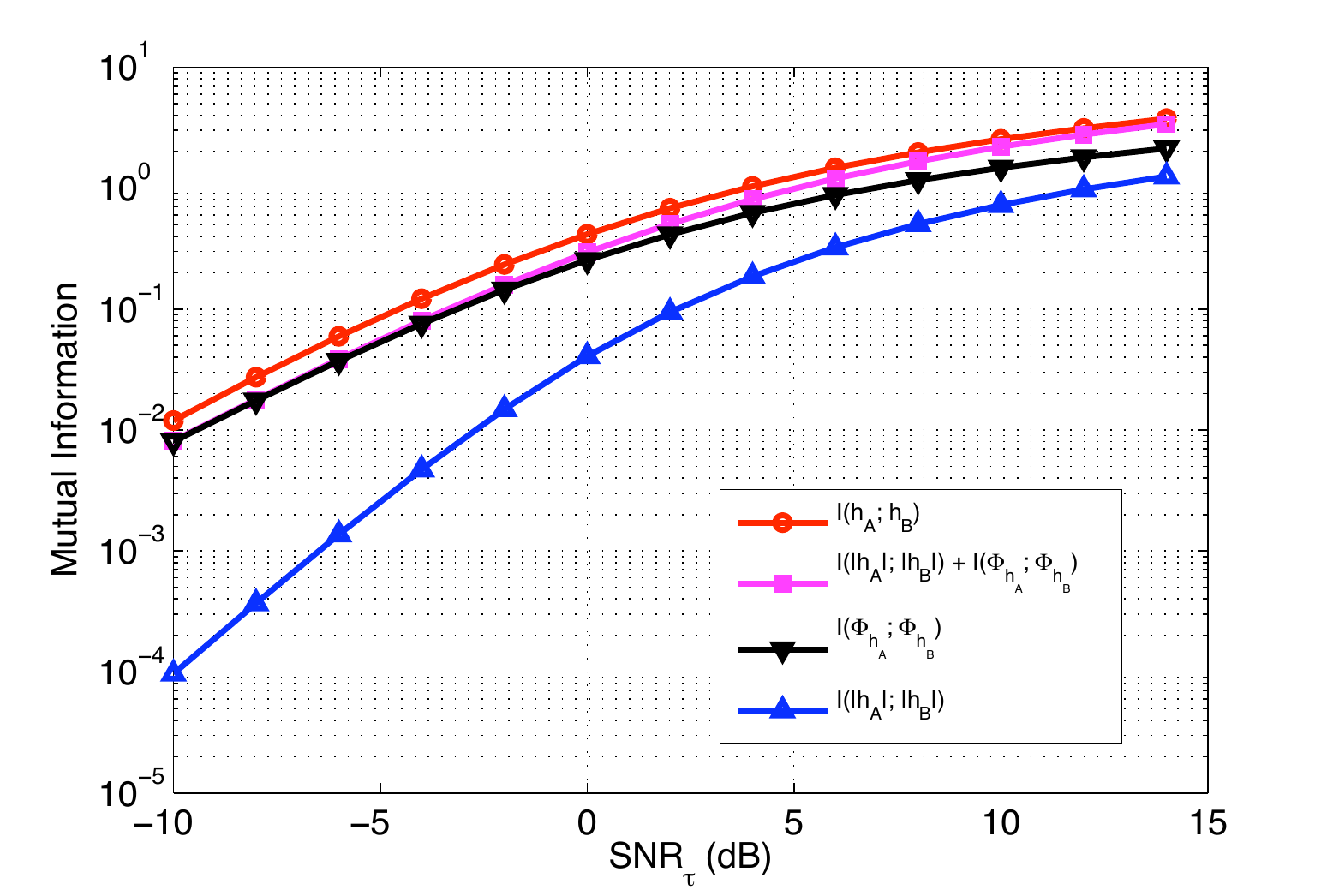,width=9cm}}
  \caption{Comparison of the secret key capacity to the sum of the
    mutual information between the magnitude of the observations and
    that between the phases.}
  \label{fig.mi_comp}
\end{figure}

\section{Design and Algorithms}
\label{sec.implement}

In this section we describe a key reconciliation system
based on low-density parity-check (LDPC) codes.  The basic idea behind
our design is the following.  Alice and Bob have correlated
observation $X_A^N$ and $X_B^N$, cf.~(\ref{eq.stackObs}), and shared knowledge of
a LDPC code.  First, Alice makes a quantized version $X_{A,Q}^N$ of
her observation $X_A^N$.  Generally $X_{A,Q}^N$ will {\em not} be a
codeword of the LDPC code, but it will always be an element of some
coset of the code.  She determines this coset by calculating the {\em
  syndrome} of her observation, which she sends to Bob.  Thus, the
syndrome is the public message in~(\ref{eq.syndromeGen}).  By itself the
syndrome reveals little information about the source since there are
so many sequences in the coset.  This is why this construction
satisfies the secrecy condition of~(\ref{eq.secrecyCond}).  However, with
knowledge of the coset and of his observation $X_B^N$, Bob can
``de-noise'' $X_B^N$ to recover $X_{A,Q}^N$ by decoding the LDPC code
with respect to the known coset in which $X_{A,Q}^N$ lies while
treating $X_B^N$ as a noisy observation of $X_{A,Q}^N$.  All cosets
inherit the distance properties of the LDPC code, which gives the
needed robustness to the random differences between $X_A^N$ and
$X_B^N$.  It should be noted that this is a well-understood method for
tackling these problems, see, e.g., \cite{Draper_SW}.  Our contribution is
really the prototyping of this system for the random source of interest
(wireless channels) and some design for
non-binary quantization.

In Sec.~\ref{sec.LDPCbackground} we provide some background on LDPC
codes and how they fit into the key generation framework of
Def.~\ref{def.secretKeySys}.  In
Sec.~\ref{sec.nonBinaryImplementation} we describe the algorithm
implemented for non-binary (four-level) quantization.

\subsection{Background, setup, and secrecy analysis}
\label{sec.LDPCbackground}

A length-$N$ rate-$R$ LDPC code over $GF(q)$ is characterized by its
$m \times N$ parity check matrix $\bP$ with elements drawn from
$GF(q)$ where $R = (1 - m/N) \log_2(q)$ bits per channel use. The
parity check matrix of a LDPC code is low-density in the sense
that the number of non-zero elements of each row is upper bounded by
some constant, regardless of the block-length $N$.  Thus, most
elements of $\bP$ are zero.  A \textit{regular} LDPC code has a
constant row-weight (number of non-zero elements) and a constant
column-weight.  An \textit{irregular} LDPC
code has a set of row and column weights, where the fraction of each
is specified by a degree distribution polynomial.

In producing $X_{A,Q}^N$ Alice has a choice of quantization.  In our
design she performs scalar quantization, quantizing each
element of $X_A^N$ independently.  Further, we study two quantization
alphabets: the first where each of the $N$ elements of $X_{A,Q}^N$ is
binary, the second where each is quaternary.  Bob may choose also to
quantize his observations prior to decoding, but there will be a loss
in information (and thus performance) if he does so.  If Bob also
quantizes his observations, he is said to perform ``hard'' decoding,
while if he does not he is said to perform ``soft'' decoding.

Alice creates her public message by multiplying her observation
$X_{A,Q}^N$ with $\bP$ to produce the length-$m$ syndrome $S^m$, $S^m = \bP X_{A,Q}^N,$ where
$m = N[1-R/\log_2(q)]$.
Within each coset there are 
$2^{N \log_2(q) (1-m/N)} = 2^{NR} \label{eq.binSize}$
sequences.  As long as $NR < I(X_{A;Q}^N; X_{B}^N)$ then recovery of
$X_{A,Q}^N$ (decoding) will be reliable.  It can be shown that, 
$H(X^N_{A,Q}|S^m)$, the uncertainty in $X_{A,Q}^N$ given the knowledge of
the public message $S^m$ is at least $NR$.  %Bounding this conditional entropy depends
%on the statistics of the source sequence.  In the analysis below we
%make the assumption that $X_{A,Q}^N$ is an i.i.d.\ sequence.  Our
%system is well modeled as such since the sampled channel coefficients
%$h_{\ell}$ are approximately independent and will be identically
%distributed after quantization, if the quantization is sufficiently
%course.
%\begin{align}
%H(X_{A,Q}^N|S^m) & = H(S^m | X_{A,Q}^N) + H(X_{A,Q}^N) - H(S^m)
%\nonumber\\& = N H(X_{A,Q}) -
%H(S^m) \nonumber\\ & \geq N H(X_{A,Q}) - m
%H(S) \label{eq.inequality}\\  &
%= NR \label{eq.uncertainty}
%\end{align}
%where the inequality in~(\ref{eq.inequality}) follows from basic
%properties of entropy.
This means that
the uncertainty in $X_{A,Q}^N$ given knowledge of the public message
$S^m$ is exactly the same as the size of the coset.  Thus, if the key
extraction function $f(\cdot)$ first quantizes $X_A^N$ to get
$X_{A,Q}^N$ and then sets the key to be equal to the index that
identifies $X_{A,Q}^N$ within the coset of the LDPC code in which it
lies, the mutual information of this index with $S^m$
will be arbitrarily small, satisfying the secrecy
condition~(\ref{eq.secrecyCond}).  Finally, since $R < (1/N) I(X_{A,Q}^N;
X_B^N)$ for successful recovery, we get the upper bound on the
achievable secrecy rate approaches~(\ref{eq.defCap}) as the
quantization gets increasingly fine.

\subsection{Design based on non-binary LDPC codes}
\label{sec.nonBinaryImplementation}

In Sec.~\ref{sec.codingSims} we present simulation results for two key generation systems based upon LDPC codes.  In the first
Alice uses binary quantization to produce $X_{A,Q}^N$ and in the
second she uses four-level quantization.  In this section we describe
the design only for the four-level (non-binary) quantization
as the one for binary quantization is based upon standard
LDPC decoding techniques.

%We construct regular LDPC codes using a code generating engine
%selected from \cite{emin} and we build irregular LDPC codes using
%conventional density evolution technique which is widely used in
%literature such as \cite{Richardson}. We also design the decoding
%algorithm of a non-binary LDPC code in $GF(4)$ based on belief
%propagation (BP). We briefly explain the BP decoding of LDPC in the
%following section.

To simplify notation, in this section we use $x_i$ to represent the
$i$th element of Alice's quantized observation $X_{A,Q}^N$ and $y_i$
to represent the $i$th element of Bob's (not necessarily quantized)
observation $X_B^N$.  As our discussion is for four-level
design, $x_i \in \{0,1,2,3\}$.  Similarly $s_i$ is the $i$th
element of $S^m$, also in $\{0,1,2,3\}$ as are all elements of $\bP$.
To simplify the design, rather than working with a code with
elements in $GF(4)$ we split each $x^N$ into bit planes, representing
each $x_i$ by a pair of binary symbols $x_{i,M}$ and $x_{i,L}$, each
taking the respective value in $\{0,1\}$ that satisfies
\begin{equation}
x_i = x_{i,L} + 2 x_{i,M}. \label{eq.binaryToFourary}
\end{equation}
We can now apply a pair of length-$N$ binary LDPC code, one to each
bit plane, or a length-$2N$ binary LDPC code to the concatenation of
the bit planes.  We implemented both and, while the latter generally
has the higher performance (though not by too much), the former allows
more flexibility (e.g., in reconstructing only the most significant
bit plane or sequential reconstruction) and is slightly simpler to
implement.  We choose to present our results on the former, using
${\cal C}_{\alpha}$, ${\bP}_{\alpha}$, and $s^{m_{\alpha}}_{\alpha}$
to represent, respectively, the code, the parity check matrix, and the
syndrome associated with $x_{\alpha}^N$ -- the sequence of concatenated $x_{i, \alpha}$
where $\alpha \in \{L,M\}$.
We note that ${\cal C}_M$ and ${\cal C}_L$ need not be the same rate
so $m_M \neq m_L$ in general.  However, in all our simulations we
choose $\bP_{M} = \bP_L$, where equality is element-wise, so ${\cal
  C}_M = {\cal C}_L$.  Recall that the syndrome is calculated as
$s^{m_{\alpha}}_{\alpha} = {\bP}_{\alpha} x_{\alpha}^N$.

To visualize the two binary LDPC codes and to see how to relate them
to the observation $y_i$ we depict the constraints involved in the key
generation process in Fig.~\ref{fig.factor_graph} using a factor
graph~\cite{Kschischang}.  The factor nodes $F_i$ constrain the
triplet of variables $(x_i, x_{i,L}, x_{i,M})$ to satisfy the
relationship of~(\ref{eq.binaryToFourary}).  In particular,
\begin{equation*}
F_i(x_i, x_{i, M}, x_{i, L}) = \left \{ \begin{array}{lll} 1, &
  \mbox{if} & x_i = x_{i, L} + 2 x_{i, M}\\ 0, &
  \mbox{else} \end{array} \right.
\end{equation*}

\begin{figure}[!t] 
  \centerline{\epsfig{figure=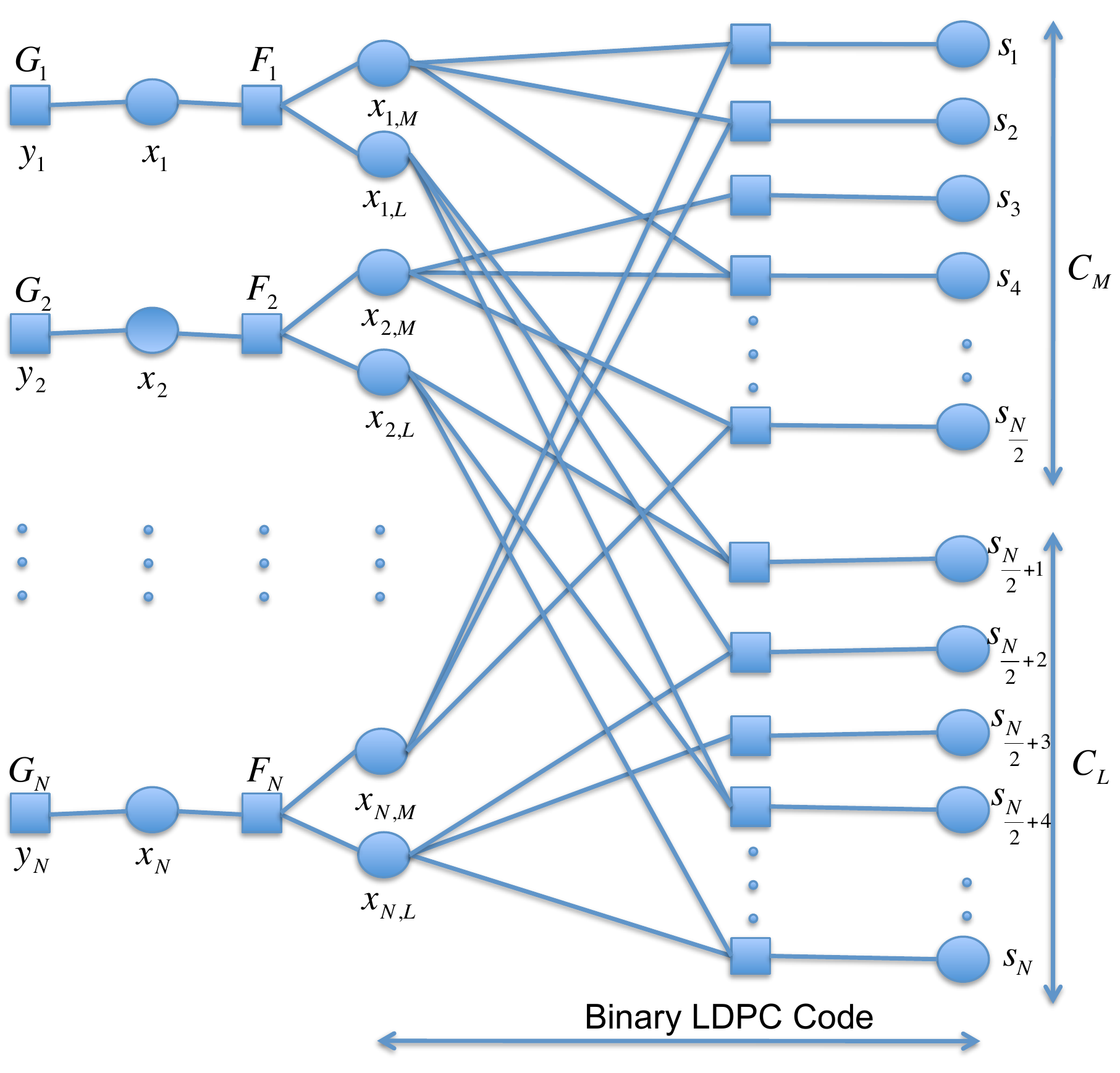,width=8cm}}
  \caption{Factor graph of 4-ary LDPC codes. Nodes $x_{i, M}$ and
    $x_{i, L}$ connect to the check nodes of different binary LDPC
    code. They are connected by local function $F_i$ which regulates
    them according to the value of $x_i$.}
  \label{fig.factor_graph}
\end{figure}

We attempt to recover $X_{A,Q}^N$, based on knowledge of $y^N$ and
$S_M^{m_M}$ and $S_{L}^{m_L}$, by using the sum-product algorithm.  In
this algorithm messages that approximate conditional probability
distributions are iteratively passed along the edges of the factor
graph.  We use the parallel message passing schedule, have all factor
nodes send messages to all variable nodes, and then vice-versa;
continuing until either the messages converge or some maximum number
of iterations is reached.  For more details of these standard aspects
of the implementation see~\cite{Kschischang}.  In the remainder of this
section we indicate the form of the message updates required for the
non-binary case.

We use the following symbols to represent the messages passed.  The
message sent from node $x_i$ to node $F_i$ and from node $F_i$ to
$x_i$ are, respectively, denoted $\mu_{x_i \rightarrow F_i}(x_i)$ and
$\mu_{F_i \rightarrow x_i}(x_i)$.  We use $\cN(p)$ to denote the
neighbors of a given node $p$. The summary operator $\sum_{\sim
  \{p\}}$ means summation over all variables \textit{except} $p$ and
the notation $\cN(p) \backslash \{q\}$ means the set of neighbors of
$p$ \textit{except} $q$. We now detail the form the general
sum-product update rules specialized for our problem.

The message passed from variable node $x_i$ to $F_i$ are calculated
as
\begin{align}
\mu_{x_i \rightarrow F_i}(x_i) &= \prod_{t \in \cN(x_i) \backslash
  \{F_i\}} \mu_{t \rightarrow x_i} (x_i) \nonumber \\ &= \mu_{y_i
  \rightarrow x_i}(x_i). \nonumber
\end{align}

As there is no marginalization at variable node $x_i$, the message
passed to $F_i$ is the same as the incoming message from $y_i$,
\begin{eqnarray}
\mu_{y_i \rightarrow x_i}(x_i) = G_i(x_i), \nonumber
\end{eqnarray}
where the local function $G_i(x_i)$ is the channel evidence
\begin{equation*}
G_i(x_i) = p_{X_i | Y_i}(x_i | y_i).
\end{equation*}

The messages passed from factor node $F_i$ to each of the two binary
variable nodes $x_{i, \alpha}$ where $\alpha \in \{M,L\}$ is
calculated as
\begin{align*}
 &\mu_{F_i \rightarrow x_{i, \alpha}} (x_{i, \alpha}) \\
 &= \frac{1}{Z}
 \!\!\sum_{\sim \{x_{i,\alpha}\}} \!\! \bigg ( F_i(x_i, x_{i, M}, x_{i,
   L}) \nonumber \!\!\!\!\!\!\!\! \prod_{t \in \cN(F_i) \backslash \{x_{i,
     \alpha}\}} \!\!\!\!\!\!\!\! \mu_{t \rightarrow F_i}(t) \bigg ),
\end{align*}
where $Z$ is a normalization factor, 
\begin{align*}
Z &= \sum_{x_{i, \alpha}} \sum_{\sim \{x_{i,\alpha}\}} \!\! \bigg ( F_i(x_i, x_{i, M}, x_{i,
   L}) \nonumber \!\!\!\!\!\!\!\! \prod_{t \in \cN(F_i) \backslash \{x_{i,
     \alpha}\}} \!\!\!\!\!\!\!\! \mu_{t \rightarrow F_i}(t) \bigg ).
\end{align*}

For example:
\begin{align*}
 &\mu_{F_i \rightarrow x_{i, M}}(1) \\
 &= \frac{1}{Z}\bigg (\mu_{x_i \rightarrow F_i}(2)
\mu_{x_{i, L} \rightarrow F_i}(0) +\mu_{x_i \rightarrow
  F_i}(3)  \mu_{x_{i, L} \rightarrow F_i}(1) \bigg ),\\
&\mu_{F_i \rightarrow x_{i, M}}(0) \\
&= \frac{1}{Z}\bigg ( \mu_{x_i \rightarrow F_i}(1)
\mu_{x_{i, L} \rightarrow F_i}(1) +\mu_{x_i \rightarrow
  F_i}(0)  \mu_{x_{i, L} \rightarrow F_i}(0) \bigg ),
\end{align*}
where $Z$ is given by
\begin{align*}
Z &= \mu_{x_i \rightarrow F_i}(2)
\mu_{x_{i, L} \rightarrow F_i}(0) +\mu_{x_i \rightarrow
  F_i}(3)\mu_{x_{i, L} \rightarrow F_i}(1) \\
  &+ \mu_{x_i \rightarrow F_i}(1)
\mu_{x_{i, L} \rightarrow F_i}(1) +\mu_{x_i \rightarrow
  F_i}(0)\mu_{x_{i, L} \rightarrow F_i}(0).
\end{align*}

The log-likelihood for $x_{i, M}$ is $\log \bigg (\frac{ \mu_{F_i
    \rightarrow x_{i, M}}(0) }{ \mu_{F_i \rightarrow x_{i, M}}(1)}
\bigg )$ which serves as the channel evidence for the binary code
${\cal C}_M$. Similarly, variable node $x_{i, L}$ calculates the
channel evidence for ${\cal C}_L$. Based on these messages, the
messages passed in the LDPC codes are standard messages where the
$s_{\alpha}^{m_{\alpha}}$ make sure that the decoding is performed
with respect to the correct cosets.  This aspect is the same as when
LDPC codes are used in Slepian-Wolf distributed source coding
problems, e.g., see~\cite{Slepian}.  The messages passed from the LDPC
codes back to the $F_i$ are $\mu_{x_{i, M} \rightarrow F_i}(x_{i, M})$
and $\mu_{x_{i, L} \rightarrow F_i}(x_{i, L})$.

Finally, the message passed from $F_i$ to variable node $x_i$ is
calculated as
\begin{align*}
&\mu_{F_i \rightarrow x_i}(x_i) \\
&= \frac{1}{Z} \sum_{\sim \{ x_i \}} \bigg (
  F_i(x_i, x_{i, M}, x_{i, L}) \!\!\! \prod_{t \in \cN(F_i) \backslash \{ x_i
    \}} \mu_{t \rightarrow F_i}(t) \bigg) \\ &= \frac{1}{Z} \!\! \sum_{\sim
    \{ x_i \}} \!\!\! \bigg ( F_i(x_i, x_{i, M}, x_{i, L}) \mu_{x_{i, M}
    \rightarrow F_i}(x_{i, M}) \mu_{x_{i, L} \rightarrow F_i}(x_{i,
    L}) \bigg), \nonumber
\end{align*}
where $Z$ is the corresponding normalization factor. For example:
\begin{equation*}
\mu_{F_i \rightarrow x_i}(2) = \frac{1}{Z}\bigg ( \mu_{x_{i, M} \rightarrow F_i}(1)
\mu_{x_{i,L} \rightarrow F_i}(0) \bigg ),
\end{equation*}
where $Z$ is
\begin{align*}
Z  &= \mu_{x_{i, M} \rightarrow F_i}(0)
\mu_{x_{i,L} \rightarrow F_i}(0) + \mu_{x_{i, M} \rightarrow F_i}(1)
\mu_{x_{i,L} \rightarrow F_i}(0) \\
&+ \mu_{x_{i, M} \rightarrow F_i}(0)
\mu_{x_{i,L} \rightarrow F_i}(1) + \mu_{x_{i, M} \rightarrow F_i}(1)
\mu_{x_{i,L} \rightarrow F_i}(1).
\end{align*}

Eventually, the messages either converge or the maximum iteration
count is reached.  In our simulations we set this maximum to $50$
iterations.  When the messages converge the marginals are computed as the following
up to a scaling factor. 
\begin{align*}
\Pr [x_i = a] &\propto \mu_{y_i \rightarrow x_i}(a) \mu_{F_i \rightarrow x_i}(a),
\end{align*}
where $a \in \{0, 1, 2, 3\}$. The algorithm sets its estimates symbol-by-symbol as
$\hat{x}_i = \arg \max_{a} \Pr[x_i = a].$

We now present the initialization of our algorithm. $GF(4)$ variable
nodes $x_i$ are initialized as:
\begin{eqnarray}
\mu_{y_i \rightarrow x_i}(x_i) = G_i(x_i), \nonumber
\end{eqnarray}
whereas $GF(2)$ variable nodes $x_{i, M}$ are initialized as:
\begin{align*}
&\mu_{F_i \rightarrow x_{i, M}} (x_{i, M}) \\
&= \frac{1}{Z} \sum_{\sim
  \{x_{i,M}\}} \!\!\! \bigg ( \! F_i(x_i, x_{i, M}, x_{i, L}) \!
\!\!\!\!\!\!\!\!\! \prod_{t \in \cN(F_i) \backslash \{x_{i, M}\}}
\!\!\!\!\!\!\!\!\!\!\! \mu_{t \rightarrow F_i}(t) \bigg ).
\end{align*}
At this step the values $x_{i, L}$ can all be set to have equal 
probability as this is the best estimate $x_{i, L}$ can be set to initially. In other words, 
the foregoing equation reduces to:
\begin{align*}
&\mu_{F_i \rightarrow x_{i, M}} (x_{i, M}) \\
&= \frac{1}{Z} \sum_{\sim \{x_{i,M}\}}
\bigg ( F_i(x_i, x_{i, M}, x_{i, L}) \mu_{x_i \rightarrow F_i}(x_i) \frac{1}{2}
\bigg ).
\end{align*}
Therefore, the message initializing $x_{i, M}$ is $\mu_{F_i
  \rightarrow x_{i, M}}(1) = \frac{1}{Z} \bigg( \frac{1}{2}\mu_{x_i \rightarrow F_i}(2)
  +\frac{1}{2}\mu_{x_i  \rightarrow F_i}(3) \bigg )$. Its corresponding log-likelihood is $\log
\bigg ( \frac{\mu_{F_i \rightarrow x_{i, M}}(0)}{\mu_{F_i
    \rightarrow x_{i, M}}(1)} \bigg )$ which serves as the initial
channel evidence in $GF(2)$ for ${\cal C}_M$. The message initializing
$x_{i,L}$ can be similarly derived which serves as the initial channel
evidence in $GF(2)$ for ${\cal C}_L$.

%We can also apply a length-$2N$ binary LDPC code to the concatenation of
%the bit planes and we implemented this alternative form of 4-ary LDPC
%code. In this form, instead of connecting to two separate binary LDPC
%codes, $x_{i, M}$ bits and $x_{i, L}$ bits all connect to the same
%binary LDPC code as shown in Fig.~\ref{fig.factor_graph2}. 

%\begin{figure}[!t] 
 % \centerline{\epsfig{figure=Figure/factor_graph2.pdf,width=9.2cm}}
  %\caption{An alternative form of the factor graph of 4-ary LDPC
   % Codes. Both $x_{i, M}$ and $x_{i, L}$ connect to the check nodes
    %of the same binary LDPC code. They are connected by local function
   % $F_i$ which regulates them according to the value of $x_i$.}
  %\label{fig.factor_graph2}
%\end{figure} 

\subsection{Phase Offset Estimation}

As discussed in Sec.~\ref{sec.PhaseNoise}, the phase offset during two-way channel training will degrade the quality of the channel measurement. Therefore, one needs to implement phase offset suppression techniques such as phase estimation \cite{WuBarness, Liang1, WuBarness2}. In this section, we present a novel approach that incorporates the estimation of phase offset into the design of reconciliation process. 

For the ease of presentation, we assume the phase offset is constant across multiple channel trainings. Our idea can easily be extended to the situation where phase offset is time varying. We propose a joint phase offset estimation and reconciliation process, formulated as:
\begin{align}
\label{eq.est}
( \hat{x}^N, ~ \hat{\theta'} ) = \!\!\!\!\!\!\! \operatorname*{arg\,max}_{\{x^N, \theta' \!~|~\! \bP x^N = s^m,  \theta' \in [0, 2\pi] \}} \!\!\!\!\!\!\!\! p_{X^N | Y^N } (x^N e^{j \theta'} | y^N e^{j \theta}).
\end{align}
Incorporating the task of phase offset estimation as in (\ref{eq.est}) into the reconciliation process puts an extra burden on the codes. To support phase offset estimation, the code rate should be lower. This reduces the cardinality of the coset that specified by the syndrome as is illustrated in Fig.~\ref{fig.coset}. The lower code rate means a lower secrecy rate, reduced by the loss in (\ref{eq.haha}). Fig.~\ref{fig.coset} depicts the coupling between this lowered rate and the joint decoding problem in (\ref{eq.est}). If the code rate is lower (a larger syndrome is used as the public message), then for a given optimal $\theta'$ parameter in (\ref{eq.est}) there will be fewer coset elements $x^N$ such that $\bP x^N = s^m$ that yields a high probability (right oval in Fig.~\ref{fig.coset}). If the original code rate is used, then there will be many high probability coset elements (left oval in Fig.~\ref{fig.coset}) and the decoding will be erroneous with high probability.  

\begin{figure}[!t] 
  \centerline{\epsfig{figure=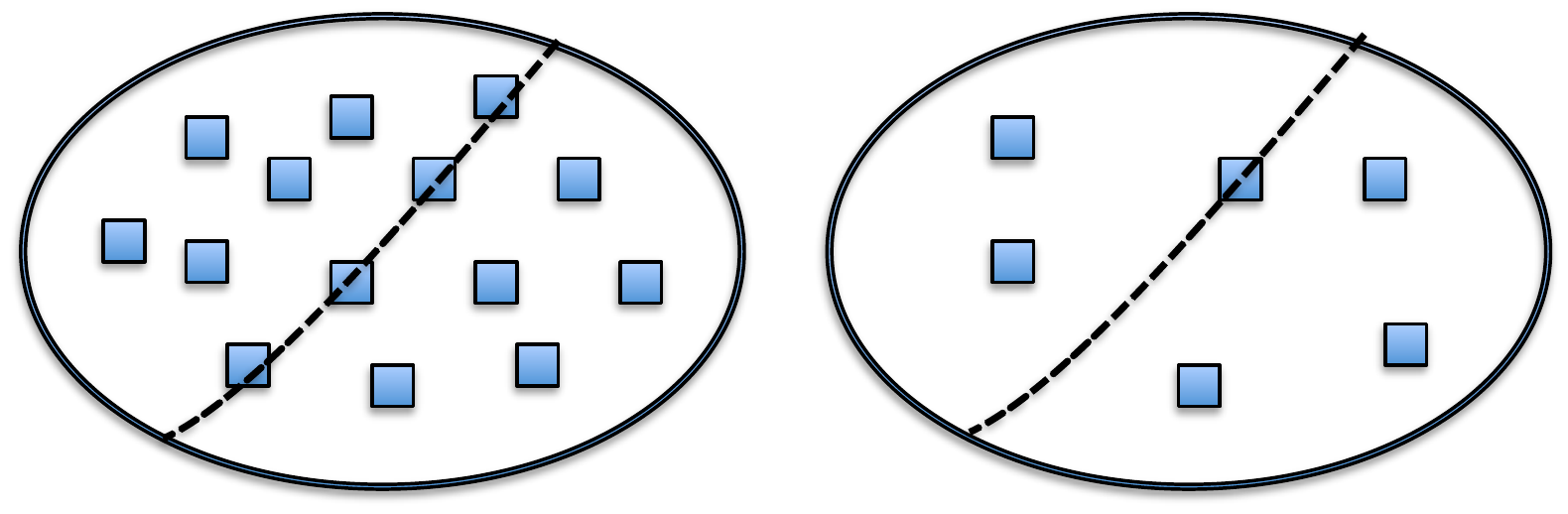,width=8cm}}
  \caption{Cosets of different LDPC code rates are shown as ovals. The one on the left corresponds to higher code rate. It contains many candidate codewords thus (\ref{eq.est}) may have non-unique solutions. The oval on the right is sparser, corresponding to a lower code rate. Thus it is possible to have a unique solution to (\ref{eq.est}). Phase rotation $\theta'$ in (\ref{eq.est}) is represented as the dashed line. }
  \label{fig.coset}
\end{figure}

While the above discussion indicates a generic approach, we now show how to integrate this search into our message passing algorithm. We propose a joint phase offset estimation and reconciliation procedure by concatenating an extra variable node $\theta$ to all the check nodes $G_i, i = 1,2,\ldots, N$, where $\theta$ denotes the random phase offset between Alice and Bob. We assume the $\theta$ is discretized such that it can only take some finite values. Then by passing message to and from variable node $\theta$, one could obtain the estimate of $\theta$. The algorithm works as follows. 

The message passed from $\theta$ to $G_i$ is
\begin{align}
\mu_{\theta \rightarrow G_i} (\theta) = \prod_{t \in \cN(\theta) \backslash G_i} \mu_{t \rightarrow \theta} (\theta).
\end{align}
The message passed from $G_i$ to $x_i$ is 
\begin{align}
\mu_{G_i \rightarrow x_i} (x_i) = \frac{1}{Z} \sum_{\theta} \bigg ( p_{X_i \theta | Y_i} (x_i, \theta | y_i) \mu_{\theta \rightarrow G_i} (\theta) \bigg ),
\end{align}
where $Z$ is some normalization factor. The message passed from $x_i$ to $G_i$ is
\begin{align}
\mu_{x_i \rightarrow G_i}(x_i) = \mu_{F_i \rightarrow x_i}(x_i).
\end{align}
Finally, the message passed back to $\theta$ is
\begin{align}
\mu_{G_i \rightarrow \theta} (\theta) = \frac{1}{Z} \sum_{x_i} \bigg ( p_{X_i, \theta | Y_i}(x_i, \theta | y_i) \mu_{x_i \rightarrow G_i} (x_i) \bigg ),
\end{align}
where $Z$ is some constant. To initialize the algorithm, one can choose the uniform distribution over all the values the $\theta$ variable can take. 

Our design extends to the cases where phase offset varies across multiple channel trainings. One can concatenate multiple $\theta$ variables, each connecting to all the check nodes $G_i$ that belong to the same channel training. The actual implementation of this algorithm is left as a future work.

\section{Simulation results}
\label{sec.results}

In this section we provide simulation results and discussion for our proposed secret key 
generation system.

\subsection{OFDM simulation results}
We first show the simulation result of an IEEE 802.11a channel. We
simulate the frequency and sampled channel coefficients and their
correlation matrices. Then we numerically 
compute the empirical secret key capacity between Alice and Bob based on our simulated time domain channel coefficients under different channel
environment.

In Table.~\ref{tb.table} we list some channel parameters for a typical rich multipath OFDM 802.11a environment \cite{80211}. Secret key capacity simulated at a particular $SNR_f$ is also listed for a  quick reference. When coherence time is small, the secret key capacity becomes large as new randomness is supplied at higher rate. However, the relationship between secret key capacity and the degree of freedom $L \approx \lceil \tau_{\max} W \rceil$ is more complicated and depends on the operating $SNR_f$. While the scaling is roughly linear in delay spread and bandwidth there are second order effects that makes the relationship more complicated. This is illustrated in Sec.~\ref{sec.mutual_info_sim}.

\small
\begin{table}[!t]
	\renewcommand{\arraystretch}{1.1}
	\caption{Channel parameters and secret key capacity}	
	\label{tb.table}
	\centering
	\begin{tabular} {| c || c |}
	\hline No. of Tones ($M$) & $52$ \\ \hline Total Bandwidth
        & $20$ MHz \\ \hline Total Data Bandwidth ($W$) & 16.25 MHz \\
        \hline Signal Duration ($T$) & $3.2$
        $\mu$s \\ \hline Carrier Frequency Spacing ($\Delta f$) &
        $312.5$ kHz \\ \hline Center Carrier Frequency ($F$) & $5.18$
        GHz \\ \hline
        Coherence Time & $100 $ ms \\ \hline
        Max Delay Spread ($\tau_{\max}$) & $800$ ns \\ \hline
        Typical Indoor Delay Spread & $40$ ns - $1$ $\mu$s \\ \hline
        Typical Outdoor Delay Spread & $1$ $\mu$s - $200$ $\mu$s \\ \hline
        Secret Key Capacity ($C$) at $20$ dB & $1040$ bits/sec \\ \hline 
	\end{tabular} \\
\end{table}

\normalsize

\subsubsection{Channel coefficients simulation}
\label{sec.channel_coeff_sim}

We consider $N_p = 300$ transmission paths and assume the $52$ tones all have
the same $SNR_f$, cf.~(\ref{freq_snr}). For simplicity, we choose the
maximum delay spread $\tau_{\max}$ to be $800$ ns so that the degree
of freedom (DoF) $L \approx \lceil \tau_{\max} W \rceil = 13$. We
reduce the redundancy in the $M = 52$ frequency domain channel
coefficients by transforming them into $13$
independent sampled channel coefficients. Over 
each coherence time, we let $\tau_k$ be drawn uniformly from 
$0$ to $\tau_{\max}$ and $\beta_k$ are independent Gaussian
random variables whose variances are related to the drawn $\tau_k$ through the exponential power-delay profile.
We generate $10^6$ independent realizations of such channel and construct the
contour plots of empirical correlation matrices of frequency domain channel
coefficients and sampled channel coefficients as shown in
Fig.~\ref{fig:total_plot}.

\begin{figure}[!t]
\centering
\subfigure[Correlation matrix of frequency channel coefficients]{
\includegraphics[scale=0.5]{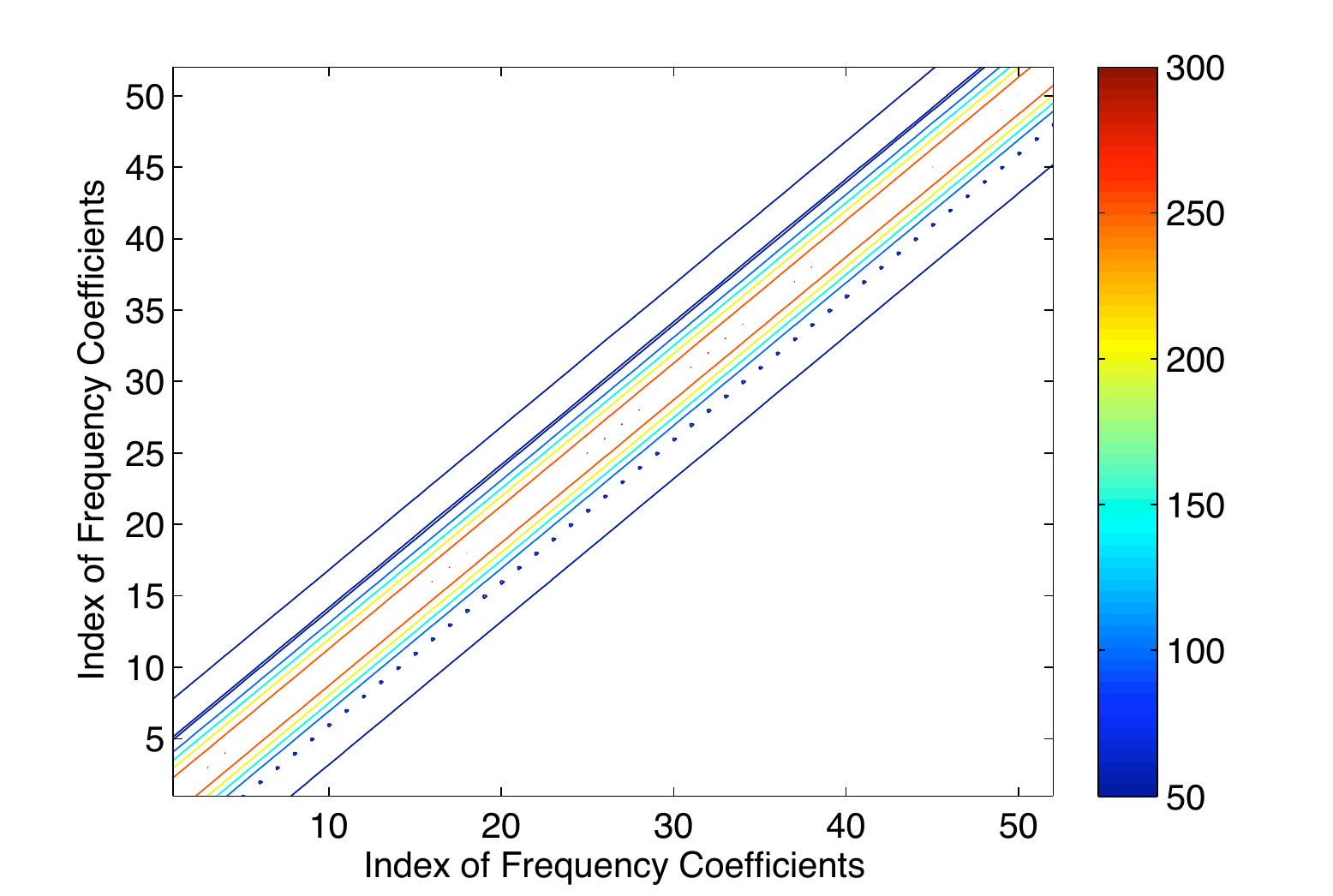}
\label{fig:subfig1}
} \subfigure[Correlation matrix of sampled channel coefficients]{
  \includegraphics[scale=0.5]{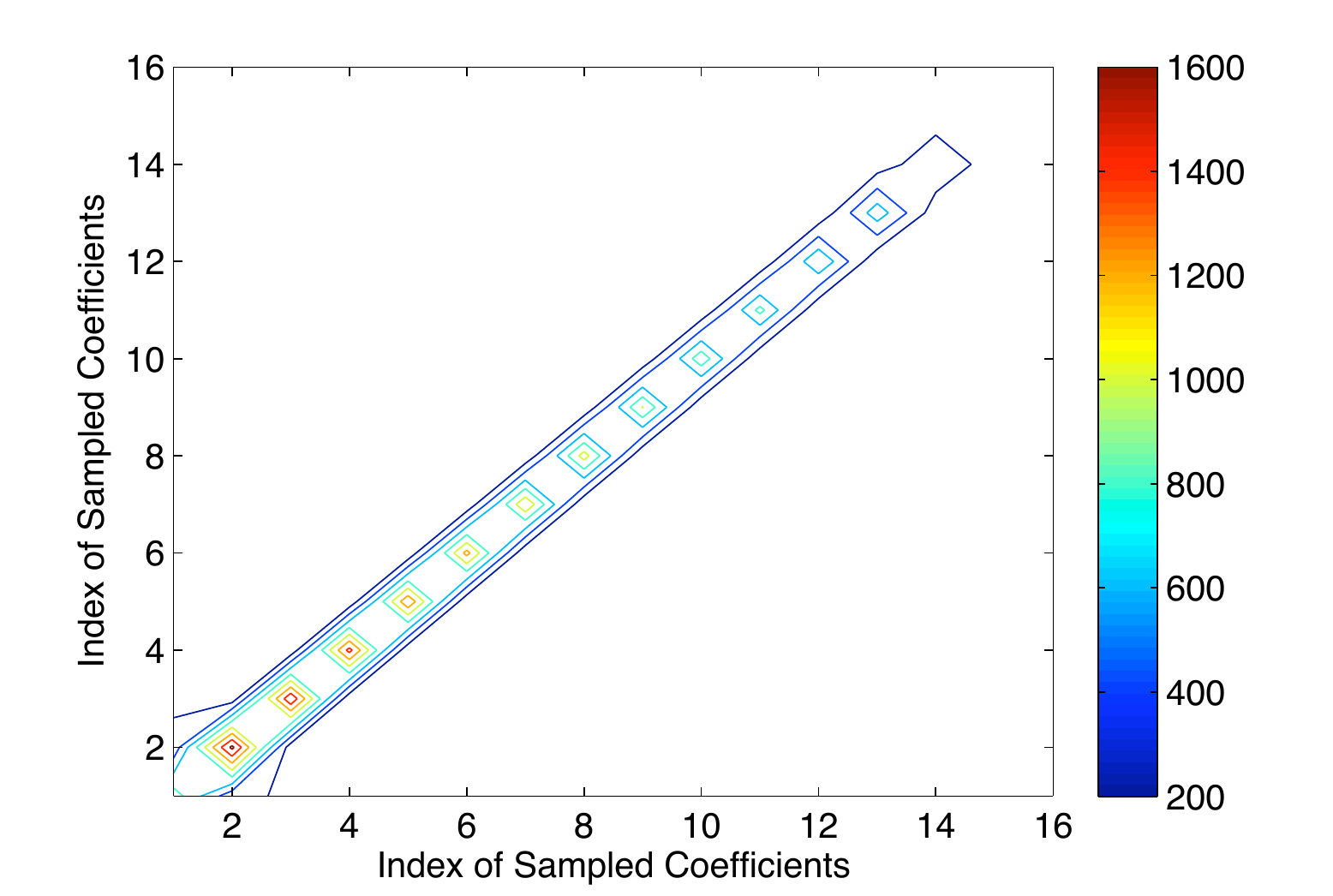}
\label{fig:subfig2}
}
%\subfigure[Variances of sampled channel coefficients]{
%\includegraphics[scale=0.5]{Figure/diag_Rtau.pdf}
%\label{fig:subfig3}
%}
%\subfigure[Eigenvalues of Correlation Matrix of Freq Channel Coefficients]{
%\includegraphics[scale=0.19]{Figure/eig.pdf}
%\label{fig:subfig4}
%}
\caption{OFDM channel coefficients simulation. Note that $13$ sampled
  channel coefficients are decorrelated from $52$ frequency domain
  channel coefficients. Note that sampled channel
  coefficients do not have the same variance.}
\label{fig:total_plot}

\end{figure}

\subsubsection{Secret key capacity simulation}
\label{sec.mutual_info_sim}
%We can characterize the secret key capacity between Alice and Bob in
%two different ways: one from the sampled channel coefficients using
%(\ref{secret_key_capacity}) and one from the frequency domain channel
%coefficients. We first calculate the capacity from the frequency domain channel
%coefficients by performing
%eigenvalue decomposition (EVD) on the correlation matrix to
%decorrelate the frequency domain coefficients. Note that this requires
%the actual statistics of frequency domain coefficients which may not
%be obtainable in a real-world implementation. Fig.~\ref{fig.3}
%provides the secret key capacity drawn from frequency domain channel
%coefficients for $L = 1, 13$ and $52$, plotted versus $SNR_f$.
%\begin{figure}[!t] 
 % \centerline{\epsfig{figure=Figure/ofdm_freqdomain_mi.pdf,width=9cm}}
  %\caption{Secret key capacity from frequency domain channel coefficients}
  %\label{fig.3}
%\end{figure} 
Secret key capacity can be computed from the first $L$ nonzero
sampled channel coefficients using (\ref{secret_key_capacity}). We
plot in Fig.~\ref{fig.2} the secret key capacity calculated from
sampled channel coefficients. Note that the secret key capacity
calculated using sampled channel coefficients is an approximation. As we see from Fig.~\ref{fig.2}, the total number of bits we can obtain per coherence time is as large as $1\times 52 \times 2 = 104$ bits at $20$ dB. It is $1040$ bits per second when coherence time is $100$ ms. 

\begin{figure}[!t] 
  \centerline{\epsfig{figure=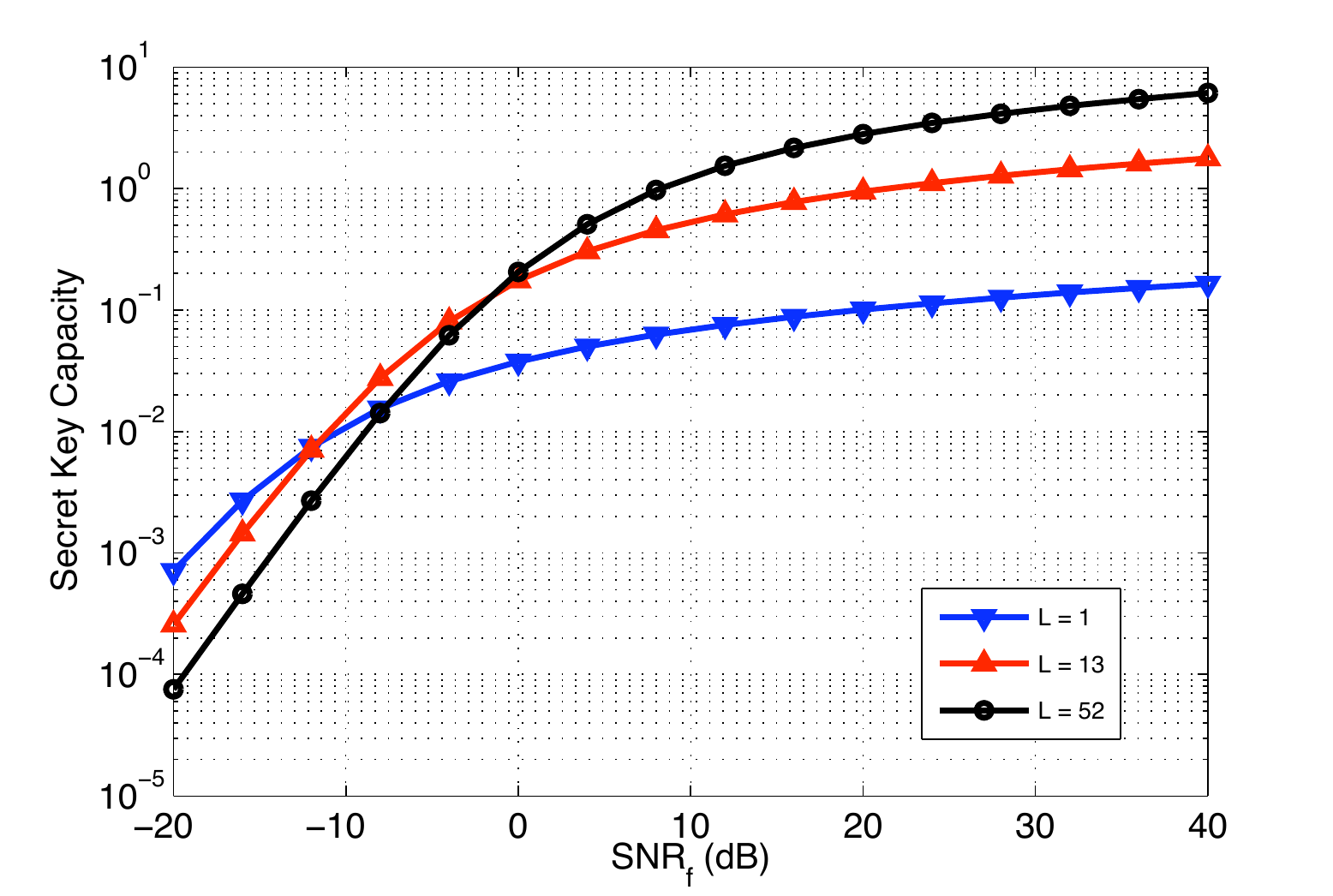,width=9cm}}
  \caption{Secret key capacity of sampled channel coefficients}
  \label{fig.2}
\end{figure}

Simulation in Fig.~\ref{fig.2} suggest that there
is no single optimal OFDM channel which has the best secret key
capacity under any $SNR_f$: under low $SNR_f$, one would like the channel to possess
fewer degree of freedom; under high $SNR_f$, one would like the channel to have
more degree of freedom. This is analogous to \cite{Chou2} where the authors 
observe that there is a trade-off between the power per degree of freedom and 
the number of degree of freedom. The intuition is also related to \cite{Verdu} 
where it is shown that peaky signal is capacity achieving input to an AWGN fading channel. 

One can also compute secret key capacity from frequency domain channel coefficients. Due to page limit constraint we omit the results. 

\subsection{LDPC Performance}
\label{sec.codingSims}

In this subsection we simulate the performance of our error correcting code. 
We allow Alice and Bob to perform multiple channel trainings. There are two ways we simulate the error correction process to reconcile
Alice and Bob's measured channel coefficients. If Bob quantizes his channel
coefficients, we term the reconciliation a \textit{hard decoding} process. On the
other hand, if Bob keeps his unquantized coefficients such that $\cY =
\bR$, we term it a \textit{soft decoding} process. In soft decoding, 
the decoder has access to Bob's full unquantized channel
coefficients which improves decoding performance. 

We let Alice and Bob perform $n =
30$ independent channel trainings yielding a block length of $N = 30 \times 52
\times 2 = 3120$. One benefit of using a large block length is that LDPC
code performs better under longer block lengths. We generate 
$400$ independent realizations of such $30$ trainings and aggregate the secret key bit error
from the channel trainings of each. For each LDPC code rate, we plot its 
corresponding SNR which yields approximately $10^{-3}$
secret key bit error rate. The number of realization is sufficient as $400 \times 3120$ is on
the order of $10^6$ which suffices to assess system performance at bit error rates of $10^{-3}$.

We simulate the performance of our error correcting code using the
sampled channel coefficients we simulated in
Section~\ref{sec.channel_coeff_sim} with $L = 13$. We connect our LDPC
simulation with the secret key capacity in
Section~\ref{sec.mutual_info_sim} by putting them in the same plot. We
plot the capacity when $L = 13$ and the performance of the binary and
non-binary (4-ary) LDPC code in Fig.~\ref{fig.ldpc_together}.
\begin{figure}[!t] 
  \centerline{\epsfig{figure=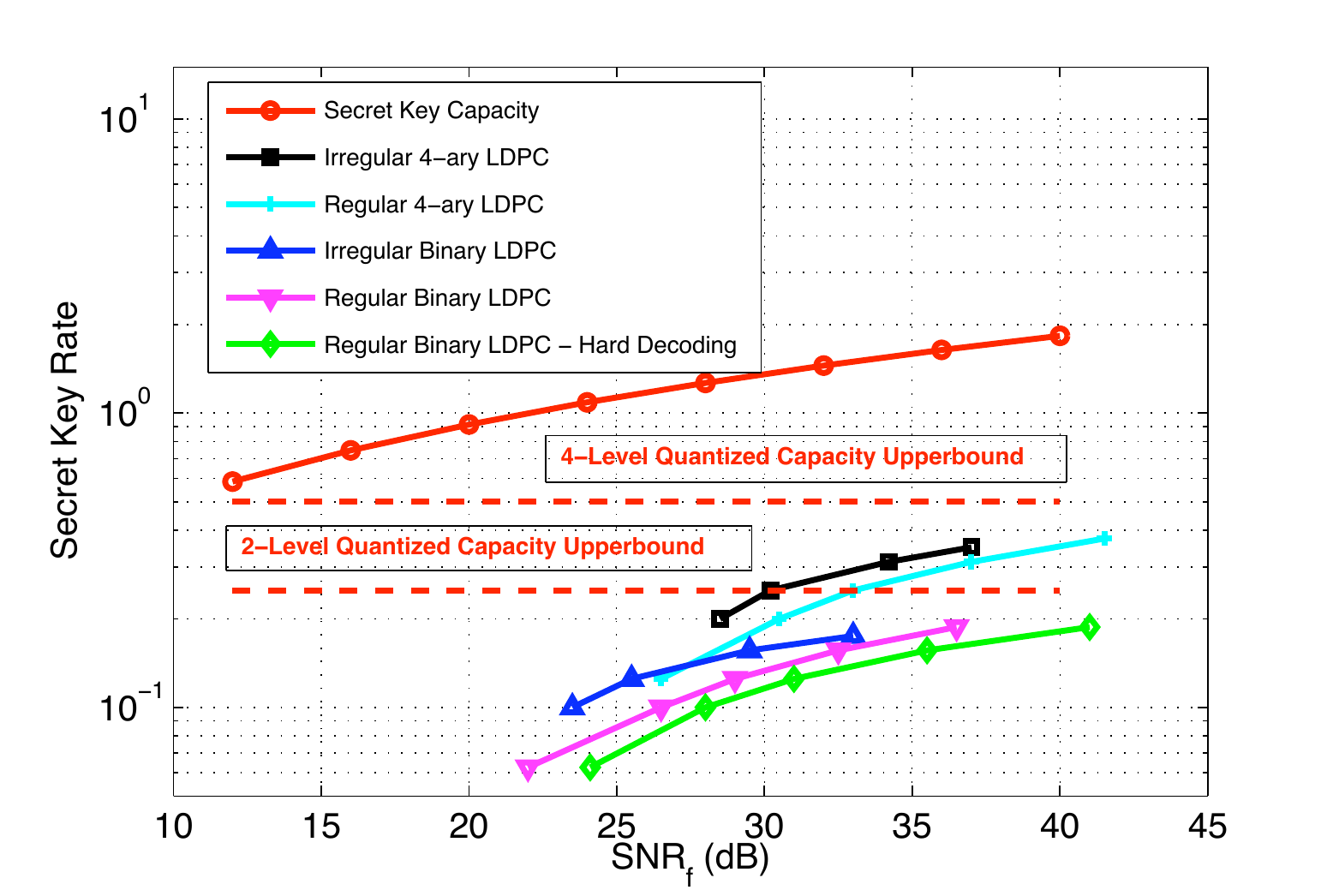,width=9cm}}
  \caption{LDPC performance}
  \label{fig.ldpc_together}
\end{figure}
The irregular LDPC codes are constructed using density evolution
technique \cite{Richardson}. We first note that our decoding
performance is improved by using soft decoding and it is further
improved by using irregular LDPC codes. Non-binary LDPC further
improves the performance and approaches the capacity at high $SNR_f$
region. LDPC codes with rate below $0.25$ are not simulated as low
code rate means less secrecy.

%\input{discussion}

% conclusion
\section{Conclusion and future work}
\label{sec.conclusion}

We study channel randomness and propose a practical system that generates secret keys from 
observing the channel randomness. We
investigate the secret key capacity shared by two end users and
find that secret key generation based on CSI is superior to the key generation based
on RSSI. This is because the CSI-based method has the larger secret key
capacity. We suggest that modern receiver circuitry should make CSI
accessible to upper layer applications. We prove that it
is always preferable to use the real and imaginary parts of the
sampled channel coefficients, as opposed to using magnitude and phase separately. Our simulation show that it is
feasible to base key generation on sampled channel coefficients. Finally, we 
implement the key generation system based both on regular and irregular LDPC codes.

\section{Acknowledgement}
We would like to thank Vincent Y.~F.~Tan, Tzu-Han Chou, Jing Yang and Xishuo Liu for their continuous discussion and support. We would also like to thank anonymous reviewers for their feedback and suggestion.

\appendix
\subsection{Proof of Theorem~\ref{thm.magPhase}}
\label{app.magPhaseProof}

We first show the following lemma.
\begin{lemma}
\label{lm.lm}
Let $X$, $Y$ and $Z$ be random variables. If $Z$ is independent \textit{either} of $X$ or of $Y$ or both, then 
\begin{equation}
I(X; Y|Z) \geq I(X; Y), \nonumber
\end{equation}
\end{lemma}
where equality holds if and only if $Z$ is independent of $(X, Y)$. 
\begin{proof}
Suppose $Z$ is independent of $X$. Follow the definition of mutual information, we have the following,
\begin{align*}
I(X; Y|Z) &= \cH(X) - \cH(X|Y, Z) \\
&\geq \cH(X) - \cH(X|Y) \\
&= I(X; Y),
\end{align*}
where $\cH(\cdot)$ denotes the differential entropy. Equality holds if $Z$ is independent of $(X, Y)$.
\end{proof} 

We now prove the theorem. 
\begin{proof}
We first prove the inequality in the theorem
\begin{align*}
&I(h_A; h_B) \\
&= I(|h_A|, e^{j\phi_A}; h_B) \\
&\stackrel{(a)}{=} I(|h_A|; h_B) + I(e^{j\phi_A}; h_B|\, |h_A|)  \\
&\stackrel{(b)}{\geq}  I(|h_A|; h_B) + I(e^{j\phi_A}; h_B) \\ 
&= I(|h_A|;|h_B|, e^{j\phi_B}) 
 + I(e^{j\phi_A};|h_B|, e^{j\phi_B}) \\ 
&\stackrel{(c)}{=}  I(|h_A|;|h_B|) + I(|h_A|; e^{j\phi_B}|\, |h_B|)  \\
&\quad + I(e^{j\phi_A};|h_B|)+ I(e^{j\phi_A}; e^{j\phi_B} | \, |h_B|)  \\ 
&\stackrel{(d)}{\geq}  I(|h_A|;|h_B|) + I(|h_A|; e^{j\phi_B} )  \\
&\quad + I(e^{j\phi_A};|h_B|)+ I(e^{j\phi_A}; e^{j\phi_B}) \\
&= I(|h_A|;|h_B|) + I(|h_A|;\phi_B )+ I(\phi_A;|h_B|)
 + I(\phi_A;\phi_B)  \\
&\stackrel{(e)}{\geq}  I(|h_A|;|h_B|)  + I(\phi_A;\phi_B), 
\end{align*}
where $(a)$ and $(c)$ follow from the chain rule of mutual information, $(b)$ follows because $|h_A|$ is independent of $\phi_A$ (cf.~Lemma~\ref{lm.lm}), $(d)$ follows because $|h_B|$ is independent of $\phi_B$  and $(e)$ follows because mutual information is non-negative. This proves the inequality in the theorem. 

The first equality in the theorem is proved by showing that the density function of $(h_A, h_B)$ can be factored into the product of density functions of $(\fRe(h_A), \fRe(h_B))$ and $(\fIm(h_A), \fIm(h_B))$. 
\end{proof}

\bibliographystyle{ieeetr}
\bibliography{./myrefs}

\end{document}